\documentclass[pdflatex,sn-mathphys]{sn-jnl}
\usepackage{enumerate}
\usepackage[shortlabels]{enumitem}
\usepackage{amsmath}
\usepackage{lmodern} 
\usepackage{graphicx}
\usepackage{subcaption}
\usepackage{fancyhdr}
\fancypagestyle{plain}{}
\DeclareMathOperator{\ORicci}{O-Ricci}

\jyear{2021}%

\theoremstyle{thmstyleone}%
\newtheorem{theorem}{Theorem}

\newtheorem{proposition}[theorem]{Proposition}% 
\newtheorem{lemma}{Lemma}

\theoremstyle{thmstyletwo}%
\newtheorem{example}{Example}%

\theoremstyle{thmstylethree}%

\begin{document}
\setcounter{page}{0}
\title[ORicci Fragility]{On the Ollivier-Ricci curvature as fragility indicator of the stock markets}

\author{\sur{Joaquin} \fnm{Sanchez Garcia} }\email{joaqsan@math.utoronto.ca}
\author{\sur{Sebastian} \fnm{Gherghe}    \email{sebastian.gherghe@mail.utoronto.ca}}
%\author{Robert McCann}


\abstract{Recently, an indicator for stock market fragility and crash size in terms of the Ollivier-Ricci curvature has been proposed in \cite{RicciFinance}. We study analytical and empirical properties of such indicator, test its elasticity with respect to different parameters and provide heuristics for the parameters involved. We show when and how the indicator accurately describes a financial crisis. We also propose an alternate method for calculating the indicator using a specific sub-graph with special curvature properties.}

\keywords{Ollivier-Ricci, curvature, finance, stock market, fragility, Wasserstein, stock correlations, optimal transport on graphs}

\maketitle
\newpage

\tableofcontents
\newpage

\section{Introduction}
Since the advent of quantitative finance, there has been an increasing need for fragility estimators which study hidden connections between entities. The 2008-2009 financial crisis is perhaps the best example of how catastrophic the consequences of a crash can be. While the mathematical and economical definition of a financial crisis varies from one reference to another, it is imperative to have good metrics and indicators to evaluate the fragility of the market in a given period. We require an estimator to:
\begin{enumerate}
    \item \label{first} Be able to capture the current state of the stock market and adapt quickly. 
    \item Take into consideration possible hidden interconnections.
    \item \label{third} Be simple but robust.
\end{enumerate}
The recent works \cite{RicciFinance} and \cite{ManyCurvaturesFinance} have proposed a very interesting object as an indicator of market fragility: the average Ollivier-Ricci curvature of a specific network. The network is constructed using the correlations of the closing stock prices, obtaining a Minimum Spanning Tree (MST) via a specific distance function and adding back ``high-value links''. In this work we further study the analytical properties and empirical results of this indicator. In particular, we seek to verify whether or not this estimator satisfies the aforementioned properties (\ref{first}-\ref{third}).

The Ollivier-Ricci ($\ORicci$) curvature of the associated network of \cite{RicciFinance} (based on \cite{MarketMining} and \cite{NetworkPerspective}) is called a ``crash hallmark''. In this document we argue that the proposed object is not an economic risk indicator but rather a good ex-post metric which can show the size and periods of a financial crisis. More specifically, we will show that the $\ORicci$ curvature of the constructed network does not indicate tendencies towards a crisis but rather accurately identifies the size and length of the crisis. 

The difference between being a crash hallmark or an economic risk indicator is fundamental in nature. The former allows us to understand the past and further improves our understanding of historical data, while the latter sends a fragility signal to the market agents prior to a potential crash. Only the indicator can be in a predictive manner, but this should not undermine the interest on the hallmark as a mathematical tool which helps us to characterize and understand crashes. \\
Analysis of the stock market using correlation networks was proposed in \cite{NetworkPerspective}. Since then, there has been a lot of interest in obtaining knowledge of the stock market from this network. In \cite{RicciFinance} we see the first appearance of curvature as an indicator for fragility. Intuitively, the deep connections between Ricci curvature bounds and entropy, recently discovered by the Lott-Villani-Sturm program (see \cite{VillaniOldAndNew}, \cite{McCann}), justify using curvature as a measure of fragility, see \cite{RicciFinance}. \\ \cite{ManyCurvaturesFinance} presents a comprehensive list of possible indicators of fragility using different definitions of curvature. We focus on analyzing analytical and numerical properties of the indicator proposed in \cite{RicciFinance}.
\section{Precise description of the algorithm} \label{AlgorithmSection}
In this section we present the algorithm for the hallmark proposed and used in \cite{RicciFinance} and \cite{ManyCurvaturesFinance}. We denote the hallmark for a given network by $\ORicci^{Net}$, which will depend on the size of the time period chosen $T$ and the edge threshold $\xi$ (see below).

Intuitively, one uses a function of correlations of data in the period of size $T$ to obtain a schematic sub-graph of the market. After a MST is chosen to be the representation of the market, edges with high correlation are added to the graph. Then, one computes the average $\ORicci$ curvature of the sub-graph using the \textit{hop distance} as the cost function (and not the weighted distance). 

\subsection{The main algorithm}\label{MainAlgorithm} 
\begin{algorithm}[H] 
To obtain the fragility indicator $(\ORicci)^{Net}$.
\begin{algorithmic}
\State Input: $T, \xi, \texttt{startDate}, \texttt{endDate}$
\For{$k \in \{1, \dots,  \texttt{endDate} - \texttt{startDate} - T \}$}
    \State Compute correlations matrix $\rho_{i,j}$ between stocks in period $[\texttt{startDate}, \texttt{startDate} + T]$.
    \State Compute cost function $D_{i,j} = \sqrt{2(1-\rho_{i,j})}$.
    \State Obtain via Prim's algorithm $G:=$ MinimumSpanningTree for the graph of all stocks with all edges, using as weights $D_{i,j}$.
    \State Add edge $xy$ to $G$ if $\rho_{x,y} \geq \xi$. 
    \State For every edge $ab$ compute the Ollivier-Ricci Curvature via \begin{equation} \label{OllivierRicci}
        \kappa(a,b) = 1 - \frac{W_1(\mu_a,\mu_b)}{d(a,b)},
    \end{equation}
    where, for a neighbor $v$ of node $a$ (i.e. $v \in N_a$), \begin{equation} \label{MeasuresFormula}
        \mu_a(v) = \frac{C_{a,v}}{\displaystyle \sum_{w \in N_a} C_{a,w}}
    \end{equation}
    and $d(a,b)$ is the (unweighted) \textit{hop distance}, counting the minimum number of steps in the shortest path of the extended graph. 
   \State Compute the average $\ORicci^{G}$ by averaging over all edges in $G$.
\EndFor
\State Return the average curvature $\ORicci^{Net}$
\end{algorithmic}
\end{algorithm}

Intuitively, for an edge $xy$ we observe $k(x,y) < 0$ if the weights and edges for neighbors of $x$ are such that it seems like $x$ is pulled away from $y$. That is, $x$ and $y$ may be connected by an edge but the weights of neighbors of $x$ away from neighbors of $y$ are relatively high. \\
In equation \eqref{OllivierRicci}, $W_1$ denotes the Wasserstein distance between probability measures, defined via \begin{equation}
  W_1(\mu,\nu) =   \inf_{\substack{\gamma \in M_{n \times n}(\mathbb{R}) \\ \gamma \times \mathbf{1}^T = \mu \\ \mathbf{1}^{T} \times \gamma = \nu}} \left\{ \sum_{i,j} d(i,j) \gamma_{i,j} \right\},
\end{equation} where again $d(i,j)$ is the \textit{hop distance}, the minimum number of steps needed to traverse the graph from verrtex $i$ to vertex $j$. The Wasserstein distance can be computed with numerical packages, like \texttt{Python Optimal Transport (POT)}, which was used in the simulations in Section \ref{sectionResults}. \\

From Algorithm \ref{MainAlgorithm}, we can formulate several questions. 
\begin{enumerate}
    \item How well does the indicator measure crisis? 
    \item What happens if one does not add high-value links? 
    \item What is the $\xi$- elasticity? How does one choose $\xi$? 
    \item As $T$ grows does the time series regularize? 
    \item As $T$ decreases do we approach white noise? 
    \item What is the impact of using the MST? Are there other better subgraphs?
    \item What happens for different distance functions functions $D_{i,j}$? 
    \item Is the average the most efficient way to measure risk? 
    \item Is it optimal to use only the hop distance to obtain the MST structure or should the hop distance be used to compute the curvature too? 
\end{enumerate}
We aim to answer all these questions throughout this work.

\section{Model dependence on parameters}
\subsection{Theory on the curvature of MST} \label{MSTSection}
\textbf{Question:} If $(G,N,E,W)$ is weighted graph and $(MST(G),N,E',W\lvert_E)$ is it's minimum spanning tree, what is the relation between $\ORicci^{G}$ and $\ORicci^{MST(G)}$? \\
\textbf{Answer} At the moment there is no apparent relation between $\ORicci$ curvature of the MST and the $\ORicci$ curvature of the original graph. We can construct simple examples with opposite results.

\begin{example} \label{OneEdge}(One edge at the time) \\
As a  toy model, consider a simple graph with 4 vertices, we add a single edge and compute the $\ORicci$ curvature. Observe that the leftmost graph is a MST for the remaining graphs. Although the connectivity of the 3 graphs is significantly different, the MST does not capture that behaviour.

\begin{figure}[h]
    \centering
    \includegraphics[scale = 0.4]{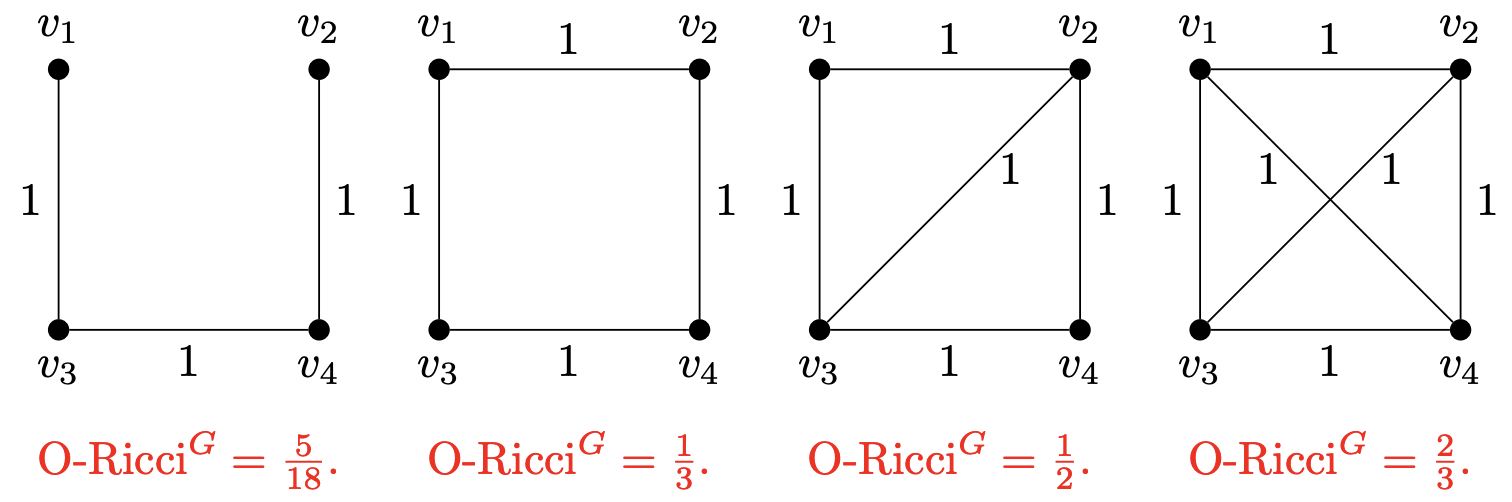}
    \caption{At each stage we add a single edge, every weight is the same, $w_{ij} = 1$, and we start from a connected graph with no cycles.}
    \label{fig:enter-label}
\end{figure}
\end{example}

\begin{example}(Very centralized node)
In the general case of a graph $G$ consisting of $n$ nodes where $n-1$ nodes are connected by one edge to one centralized node, one can immediately see that if $j \neq k \neq 1$, 
\begin{equation}
    k(v_1, v_j) = 1 - 1 = 0, \quad k(v_j,v_k) = 1.
\end{equation}
Thus there are $n-1$ edges with curvature $0$ and the remainder, of which there are $\frac{n(n-1)}{2} - (n-1)$, have curvature $1$. Therefore
\begin{equation}
    \ORicci^{G} = \frac{n-2}{n}.
\end{equation}

\begin{figure}[h]
    \centering
    \includegraphics[scale = 0.3]{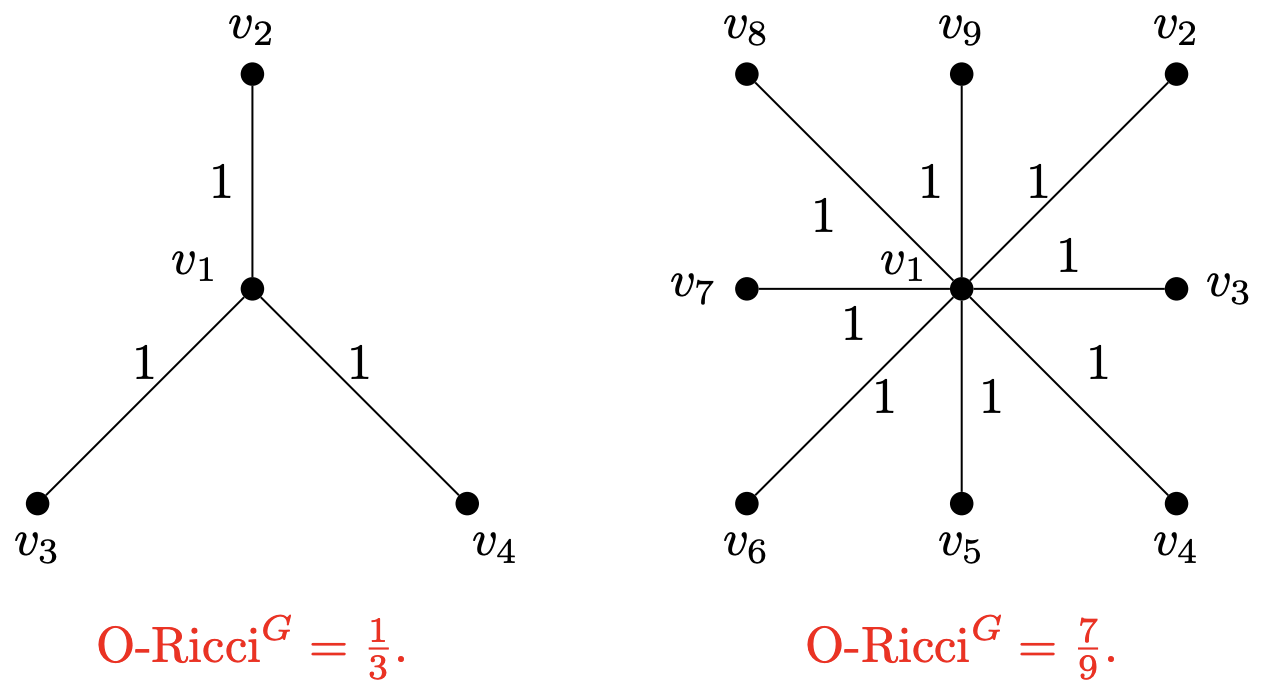}
    \caption{$v_1$ is a completely centralized node on which all others depend.}
    \label{fig:CentralizedNode}
\end{figure}
\end{example}

\begin{example}(Complete Graph) \label{example_completegraph}
In fact, we can compute the average $\ORicci$ curvature in the general case of a complete (fully connected) graph as follows. For $j \neq k$ we have $W_1(\mu_j, \mu_k) = \frac{1}{n-1}$, so $k(v_j, v_k) = 1 - \frac{1}{n-1} = \frac{n-2}{n-1}$. The curvature of each edge is the same, so
\begin{equation}
    \ORicci^{G} = \frac{n-2}{n-1}.
\end{equation}
\end{example}

From Example \ref{OneEdge}, it becomes apparent that the $\ORicci$ curvature of the Minimum Spanning Tree (MST) and the original graph are not directly correlated. Consequently, in Algorithm \ref{MainAlgorithm}, the decision to measure fragility using the MST represents a significant choice. Moreover, it remains unclear whether the edges removed during the MST selection process significantly contribute to fragility. To address this uncertainty, we consider two potential alternatives:
\begin{enumerate}
\item Computing the $\ORicci$ of the complete graph $G$ with the weights $D_{i,j}$, though this approach is computationally very costly.
\item Opting for a different sub-graph $G' \subset G$ instead of the MST that better captures curvature inheritance (see Section \ref{ORGSS}).
\end{enumerate}

\subsection{Adding or removing an edge}

Consider an undirected weighted graph $G := (N,E,W)$, where $N$ is the set of nodes, $E$ the set of edges and $W = \{ w_{e}: e \in E\}$ is the set of weights for all edges. Assume that for $x,y \in N$, $xy \not \in E$, and let us define $G^* = (N, E^*, W^*)$ where $E^* = E \cup \{xy\}$ and $W^* = W \cup \{w_{x,y}\}$. In other words, $G^*$ the graph $G$ modified by adding the edge $e := xy$ joining $x$ and $y$ with weight $w_{x,y}$. 

Let $d$ denote the Hop distance in $G$ and $d^*$ denote the Hop distance in $G^*$. Observing that adding an edge can only decrease the number of steps between nodes, for every $a\neq b \in N$ it follows that
\begin{equation} \label{distancesbound}
    d^*(a,b) \leq d(a,b).
\end{equation}
Keeping the notation of Algorithm \ref{OllivierRicci}, for every $a\neq b\in N$ we define the measures $\mu_a(\cdot)$ and $\mu_b(\cdot)$ by
\begin{equation} \label{measuresFormulas}
    \mu_{a}(b) = \frac{w_{a,b}}{\displaystyle \sum_{v \in N_a} w_{a,c}}, \: \:  \mu^*_{a}(b) = \frac{w^*_{a,b}}{\displaystyle \sum_{v \in N^*_a} w^*_{a,c}},
\end{equation}
where $N_a$ is the set of neighbors of $a$ in $G$ and $N^*_a$ is the set of neighbors of $a$ in $G^*$. The difference between $N_a$ and $N_a^*$ relies only on the condition $a \in \{x,y\}$, as 
\begin{equation}
   N_a^* =  \begin{cases}
    N_a, & \text{ if } a \not \in \{x,y\}, \\
    N_a \cup \{x \}, & \text{ if } a = y, \\
    N_a \cup \{y \}, & \text{ if } a = x.
    \end{cases}
\end{equation}
As a result, the measures in \eqref{measuresFormulas} only change if $a \in \{x,y\}$, or equivalently $\mu_a = \mu^*_a$ if $ a \not \in \{x,y\}$. Even if the measure $\mu^*_a$ remains unchanged from $\mu_a$, the Wasserstein distance $W^{d^*}$ may change from $W^d$ as the distance function has changed.

Denote by $\ORicci^{G}$ (respectively $\ORicci^{G^*}$) the Ollivier-Ricci curvature computed with the induced Hop distance $d$ on $G$ (respectively $d^*$ on $G^*$), as described in Section \ref{AlgorithmSection}.\\

Our first result, Proposition \ref{UpperBoundCurvature}, measures the total change in edge-curvature in terms of the relative error between $d$ and $d^*$. As we add a single edge to the graph, the only control on the curvature depends on how much the hop distances change.

\begin{proposition} \label{UpperBoundCurvature} (Bound on single-edge curvature change) \\
For any nodes $a,b \in N$, the difference in the $\ORicci$ curvature $k(a,b)$ of $G$ and $k^*(a,b)$ of $G^*$ is bounded, with upper bound
 \begin{align} \label{FirstBoundCurvature}
   & k^*(a,b) - k(a,b) \leq \frac{W^{d}_1(\mu_a,\mu_b) - W^{d^*}_1(\mu^*_a,\mu^*_b)}{d^*(a,b)}.
\end{align}
Moreover, if $a,b \not \in \{x,y\}$ and there exists a $\pi^*$ minimizing $W^{d^*}(\mu^*_a,\mu^*_b)$, then 
\begin{align}
   &k^*(a,b) - k(a,b) \leq \frac{\lvert \lvert d - d^* \rvert \rvert_{\infty,N}}{d^*(a,b)}, \label{secondBound}
\end{align}
where 
\begin{equation}
    \lvert \lvert d(a,b) - d^*(a,b) \rvert \rvert_{\infty,N} = \max_{(u,v) \in N^2} \lvert d(u,v) - d^*(u,v) \rvert,
\end{equation}
and $\mu_a,\mu_b,\mu_a^*,\mu_b^*$ are as in \eqref{measuresFormulas}. 
\end{proposition}
\begin{proof}
Directly from the definition of $k(a,b)$ in \eqref{OllivierRicci},
\begin{align}
    k^*(a,b) - k(a,b) & = \frac{W_1^d(\mu_a,\mu_b)}{d(a,b)} - \frac{W^{d^*}(\mu^*_a,\mu^*_b)}{d^*(a,b)} \nonumber \\
    & = \frac{d^*(a,b)W^{d}_1(\mu_a,\mu_b) - d(a,b)W^{d^*}_1(\mu^*_a,\mu^*_b)}{d(a,b)d^*(a,b)}. \label{SecondCurv}
\end{align}
Plugging the inequality \eqref{distancesbound} into \eqref{SecondCurv} yields \eqref{FirstBoundCurvature}. To obtain inequality \eqref{secondBound}, note that if $\pi^*$ is a minimizer for $W^{d^*}(\mu_a,\mu_b)$, then 
\begin{align}
    &W^d(\mu_a,\mu_b) - W^{d^*}(\mu_a,\mu_b) \nonumber \\
    &= \inf_{\pi \in \Gamma(\mu_a,\mu_b)} \int d(x,y) d\pi(x,y) - \int d^*(x,y) d\pi^*(x,y) \nonumber \\
    & \leq \int d(x,y) - d^*(x,y) d\pi^*(x,y) \nonumber \\
    & \leq \lvert \lvert d - d^* \rvert \rvert_{\infty,N}. \label{dinftyestimate}
\end{align}
Inequality \eqref{secondBound} follows from plugging in \eqref{dinftyestimate} into \eqref{FirstBoundCurvature}.
\end{proof}

Unless one applies more conditions to the graphs $G$ and $G^*$, we expect the above inequalities to be sharp, based on the following example. 

\begin{example}[Inequalities \eqref{FirstBoundCurvature} and \eqref{secondBound} are attained for every vertex pair.] 
Let $G^*$ denote a complete graph of $n$ vertices $\{v_1, \dots, v_n\}$ with all edge weights equal to $1$, and let $G$ be the graph created by removing the edge $\{v_1 v_2\}$ from $G^*$. In this case, a quick computation shows that the estimates \eqref{FirstBoundCurvature} and \eqref{secondBound} are attained for every possible vertex pair. \\

Recalling Example \ref{example_completegraph}, we know that each edge of the graph $G^*$ has the same curvature $k^*(v_i,v_j) = \frac{n-2}{n-1}$ and hence this is $\ORicci^{G^*}$. In particular, $W^{d^*}_1(\mu_i, \mu_j) = \frac{1}{n-1}$. Then, it is easy to see that $W_1(\mu_1, \mu_2) = 0$ and hence $k(v_2,v_1) = 1$. Since $G^*$ is complete, $d^*(v_1,v_2) = 1$. We can now easily verify that both sides of the inequality (\ref{FirstBoundCurvature}) are $-\frac{1}{n-1}$. \\
After removing one edge from the complete graph $G^*$, the curvature of all unaffected vertices is unchanged,
\begin{equation}
    k(v_j, v_k) = k^*(v_j, v_k) = \frac{n-2}{n-1}, \quad j \neq k \text{ and } j,k \neq 1,2,
\end{equation}
and there are $\frac{(n-2)(n-3)}{2}$ such edges. In this case both sides of both inequalities (\ref{FirstBoundCurvature}) and (\ref{secondBound}) are $0$. \\
It remains to compute the curvature $k(v_j, v_1)$ for $j \geqslant 3$ and $k(v_k, v_2)$ for $k \geqslant 3$ (we already computed above that $k(v_2, v_1) = 1$ as $W_1(\mu_1, \mu_2) = 0$). There are $2(n-2)$ such cases and by symmetry they are all identical, so without loss of generality consider $k(v_3, v_1)$. To compute $W_1(\mu_3,\mu_1)$, we distinguish two possibilities: sending mass from $\delta_2$ to the vertices $v_3, \dots, v_n$ equally which has cost $\frac{1}{n-1}$, and sending all $\frac{1}{n-1}$ mass from $v_1$ to $v_3$ which also has cost $\frac{1}{n-1}$. Thus $W_1^d(\mu_3,\mu_1) = \frac{2}{n-1}$, and hence $k(v_3, v_1) = \frac{n-3}{n-1}$. Since clearly $d^*(v_3,v_1) = 1$, the both sides of the inequality (\ref{FirstBoundCurvature}) are $\frac{1}{n-1}$. \\ 
Now, we can compute
\begin{equation}
    \ORicci^G = \frac{n^2(n-3)+1}{n(n-1)}.
\end{equation}
Compare this with example \ref{example_completegraph}, where $\ORicci^{G^*} = \frac{n-2}{n-1}$.
\end{example}

Our next proposition explores the relationship between the curvature and the connectivity of the graph via the degrees of the affected nodes. In what follows, let $n_a = \#N(a)$ denote the degree of a node $a$, and recall that by $\{xy\}$ we specifically denote the edge added to the graph $G$ to obtain the graph $G^*$ (and by $x, y$ we denote the affected nodes). We first establish the following lemma, which features the key idea (see Figure \ref{fig:map}). 

\begin{lemma}(Change in affected nodes) \\
For the affected nodes $x,y$, we have
\begin{equation} \label{boundmus}
    W^{d}(\mu_x, \mu_x^*) \leq \frac{1}{n_x+1}.
\end{equation}    
\end{lemma}
\begin{proof}
By the definition \eqref{measuresFormulas},
\begin{equation}
    \mu_x = \frac{1}{n_x}\sum_{b \in N(x)} \delta_b, \hspace{0.5cm} \mu_x^* = \frac{1}{n_x+1} (\sum_{b \in N(x)}\delta_b + \delta_y),
\end{equation}
so the map which keeps $\frac{1}{n_x+1}$ in each vertex different than $y$ and transports $\frac{1}{n_x} - \frac{1}{n_x+1}$ onto $y$ is an admissible map, see Figure \ref{fig:map}. Hence,
\begin{equation}
     W^{d}(\mu_x, \mu_x^*) \leq \frac{1}{n_x(n_x+1)} \sum_{a \in N(x)} d(x,a) = \frac{n_x}{n_x(n_x+1)} = \frac{1}{n_x+1}.
\end{equation}
\end{proof}

\begin{figure}[H]
    \centering
    \includegraphics[scale = 0.4]{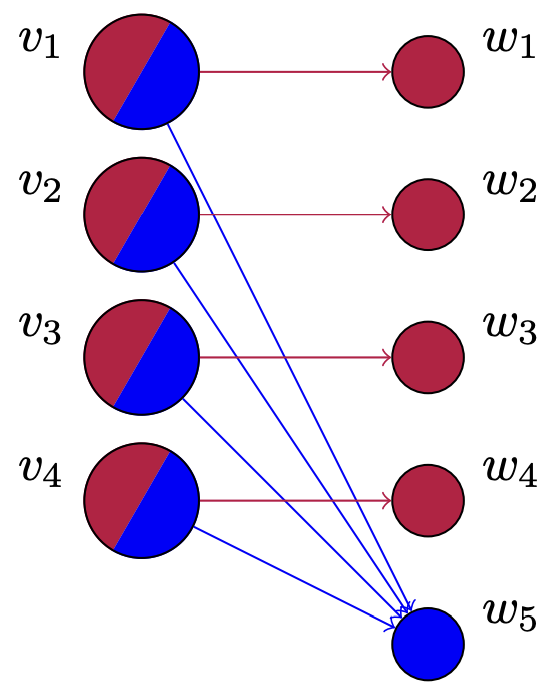}
    \caption{The admissible transport map between $\mu_x$ and $\mu_x^*$. When adding an edge connecting $x$ to another node $y$ we can keep all of the previous mass and transport the excess to the new neighbor $y$ of $x$.}
    \label{fig:map}
\end{figure}

Since we do not take the graphs $G$ and $G^*$ to be directed, the argument for $\mu_x$ and $\mu_x^*$ is also valid for $\mu_y$ and $\mu_y^*$.

\begin{proposition} \label{firstbound}(Bound on curvature in terms of degrees) \\
For the added edge $xy$, we have 
\begin{align} 
    &k^*(x,y) - k(x,y) \nonumber \\
    &\leq \frac{1}{d^*(x,y)} \left( W^d(\mu_x^*,\mu_y^*) - W^{d^*}(\mu_x^*,\mu_y^*) + \frac{1}{n_x+1} + \frac{1}{n_y+1} \right). 
\end{align}
Furthermore, if there exists a minimizer $\pi^*$ for $W^{d^*}$, then
\begin{equation} \label{BoundWithDegrees}
    k^*(x,y) - k(x,y) \leq  \frac{1}{d^*(x,y)} \left(  \lvert \lvert d-d^* \rvert \rvert_{\infty,N} + \frac{1}{n_x+1} + \frac{1}{n_y+1}\right).
\end{equation}
\end{proposition}

\begin{proof}
By the triangle inequality, 
\begin{equation}
    W^d(\mu_x,\mu_y) \leq W^{d}(\mu_x,\mu_x^{*}) + W^{d}(\mu_x^{*},\mu_y^{*}) + W^{d}(\mu_y^{*},\mu_y)
\end{equation}
The result then follows from the previous lemma. The second statement is obtained as in the proof of Proposition \ref{UpperBoundCurvature}.
\end{proof}

Note that because $d^*(x,y) \geq 1$, we immediately get a more familiar type of bound: 
\begin{equation} \label{BoundWithDegreesAndOne}
    k^*(x,y) - k(x,y) \leq  \lvert \lvert d-d^* \rvert \rvert_{\infty,N} + \frac{1}{n_x+1} + \frac{1}{n_y+1}.
\end{equation}
Equation \eqref{BoundWithDegreesAndOne} shows a direct connection between connectivity and curvature, via the right hand-side which involves the degrees of $x$ and $y$.

%\subsection{Adding high-value links} \label{XiSection}
%In Algorithm \ref{MainAlgorithm}, we obtain the MST and then add back edges satisfying
%\begin{equation} \label{ThresholdCondition}
%     \rho_{i,j} \geq \xi.
%\end{equation}
%If $D_{i,j} = h(1-\rho_{i,j})$, where $h: \mathbb{R} \to \mathbb{R}$ is positive and continuous function with $h(0) = 0$, then \eqref{ThresholdCondition} adds the edges $ij$ whose corresponding weights $D_{i,j}$ are small. From this observation, we expect $\mu^*$ to be similar to $\mu$ due to the stability result of \cite{StabilityOfMST}. The exact nature of this similarity remains an open question. 
%\newpage
\subsection{Time parameter and smoothing}
In practice, one does not want to wait too long to compute good indicators of the dynamical properties of the market. However, for an indicator to demonstrate relative consistency, it must exhibit clearer trends with increased data availability.
\begin{figure}[H] % "[t!]" placement specifier just for this example
\begin{subfigure}{0.48\textwidth}
\includegraphics[width=\linewidth]{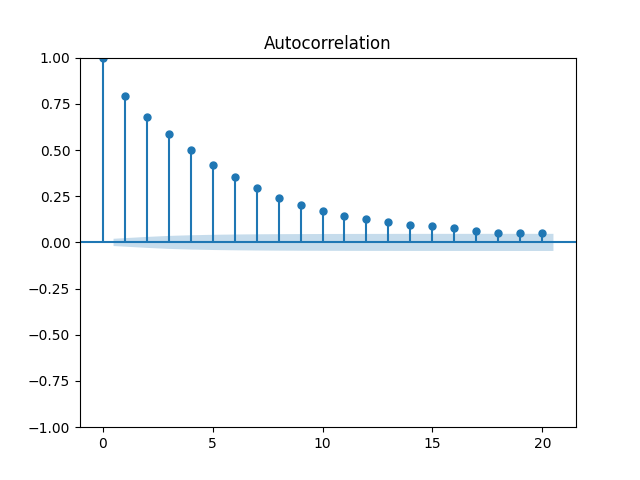}
\caption{$T = 22$} \label{fig:a}
\end{subfigure}\hspace*{\fill}
\begin{subfigure}{0.48\textwidth}
\includegraphics[width=\linewidth]{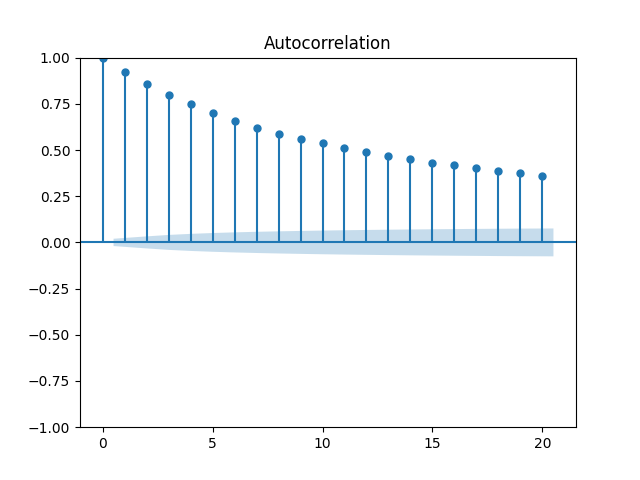}
\caption{$T = 132$} \label{fig:b}
\end{subfigure}

\medskip
\begin{subfigure}{0.48\textwidth}
\includegraphics[width=\linewidth]{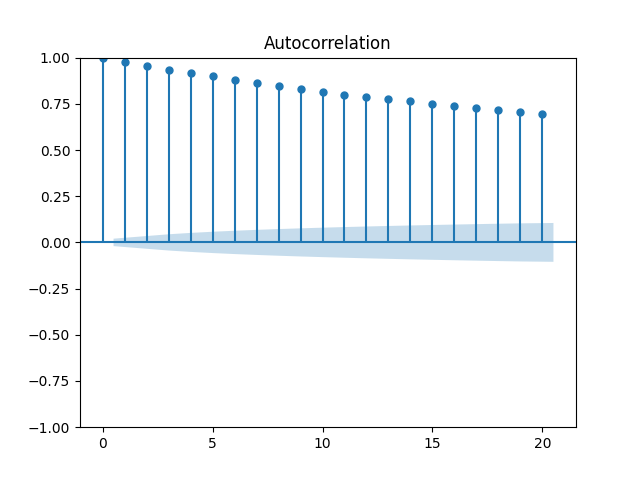}
\caption{$T = 500$} \label{fig:c}
\end{subfigure}\hspace*{\fill}
\begin{subfigure}{0.48\textwidth}
\includegraphics[width=\linewidth]{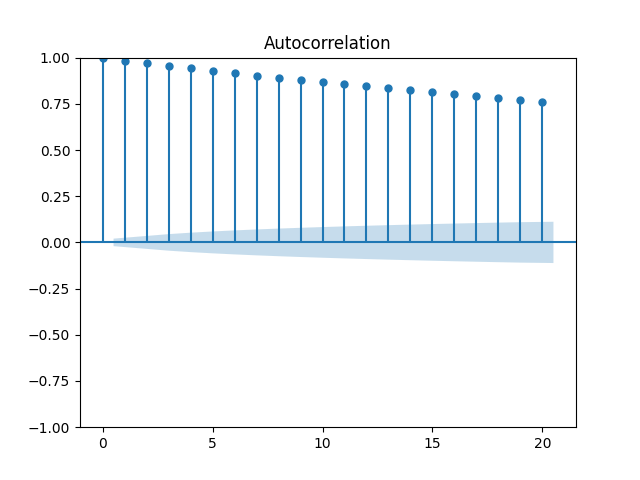}
\caption{$T = 700$} \label{fig:d}
\end{subfigure}

\medskip
\begin{subfigure}{0.48\textwidth}
\includegraphics[width=\linewidth]{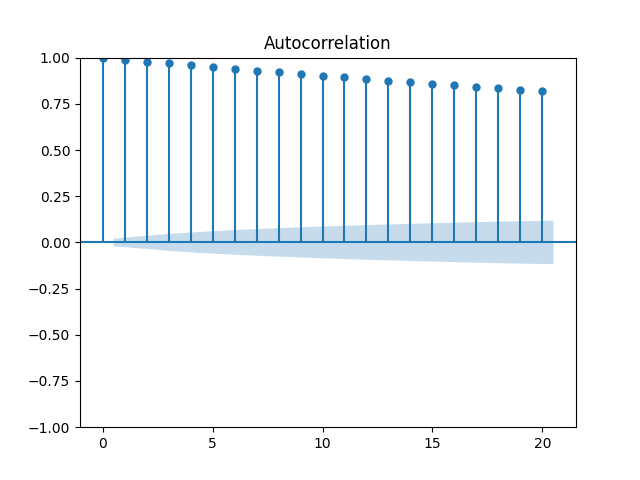}
\caption{$T = 1000$} \label{fig:e}
\end{subfigure}\hspace*{\fill}
\begin{subfigure}{0.48\textwidth}
\includegraphics[width=\linewidth]{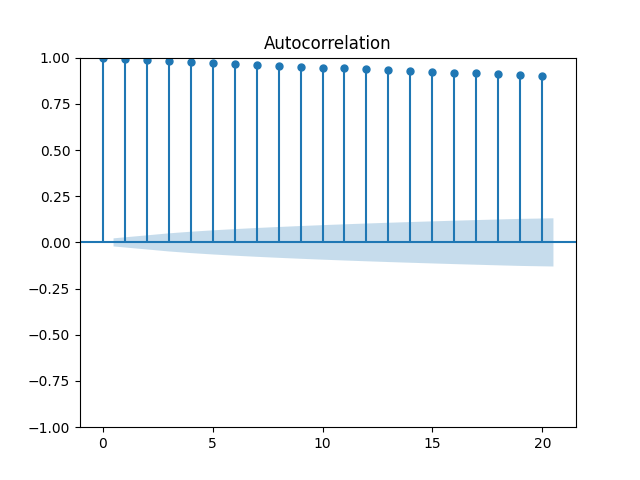}
\caption{$T = 2000$} \label{fig:f}
\end{subfigure}
\caption{Behaviour of auto-correlation in large $T$.} \label{fig:LargeT}
\end{figure}

In Figure \ref{fig:LargeT} we observe that (as expected) the indicator shows a stronger auto-correlation as $T$ increases. As a measure of economical ex-post analysis, higher values of $T$ are preferred.

\subsubsection{Statistical Test: Significance vs White Noise} \label{TSection}
We now show via auto-correlation plots that in the $ T \to 0$ limit the indicator behaves increasingly as white noise. 
%%%%%%%%%%%%%%%%%%%%%%%%%%%%%%%%%%%%%%%%%%%%%%%%%%%%%%%%%%%%5 LET ME DO EVEN CLOSER TO 0
\begin{figure}[H] % "[t!]" placement specifier just for this example
\begin{subfigure}{0.48\textwidth}
\includegraphics[width=\linewidth]{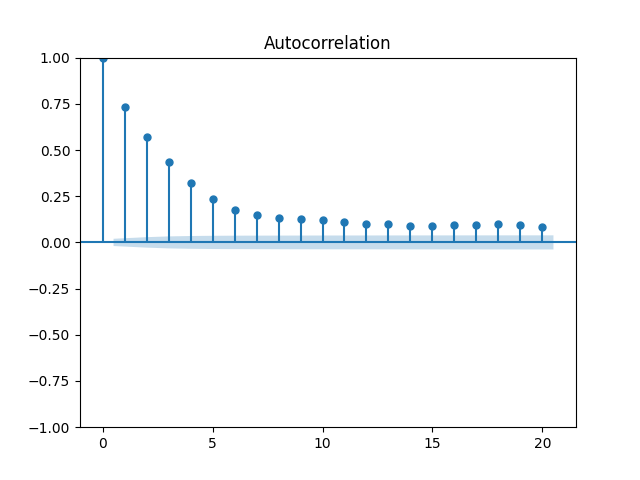}
\caption{$T = 5$} \label{fig:a2}
\end{subfigure}\hspace*{\fill}
\begin{subfigure}{0.48\textwidth}
\includegraphics[width=\linewidth]{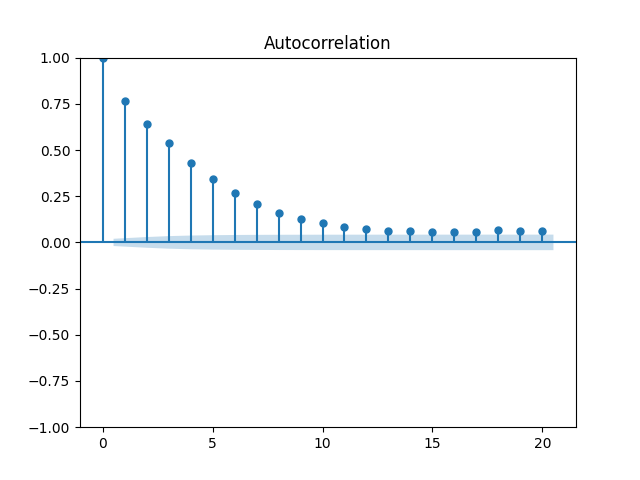}
\caption{$T = 10$} \label{fig:b2}
\end{subfigure}

\medskip
\begin{subfigure}{0.48\textwidth}
\includegraphics[width=\linewidth]{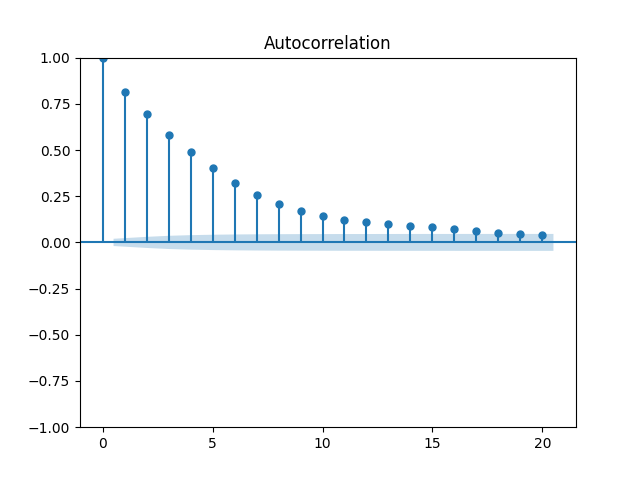}
\caption{$T =20$} \label{fig:c2}
\end{subfigure}\hspace*{\fill}
\begin{subfigure}{0.48\textwidth}
\includegraphics[width=\linewidth]{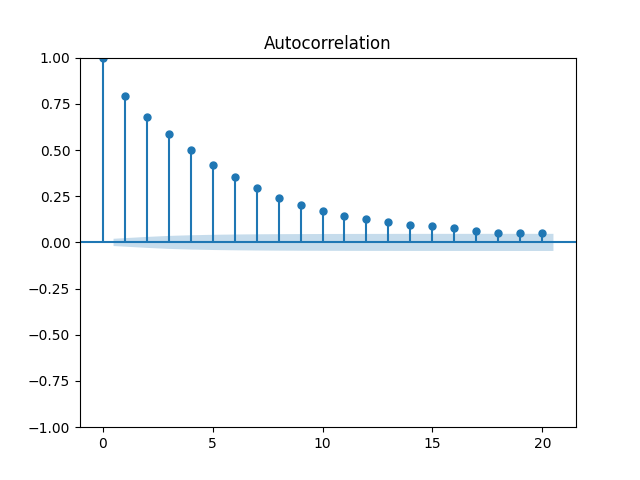}
\caption{$T = 22$} \label{fig:d2}
\end{subfigure}
\caption{Smaller values of $T$} \label{fig:SmallTFig}
\end{figure}

From Figure \ref{fig:SmallTFig} it is clear that taking small values of $T$ for the computation of $\ORicci^{Net}$ may result in short-term mistakes in trends. These mistakes are fundamental when trying to use $\ORicci^{Net}$ as a measure of fragility.

\subsection{Different metric to obtain MST and the choice of the ultra metric} \label{MetricsSection}
In \cite{StabilityOfMST} it was shown that the MST obtained from $D_{i,j} = \rho_{i,j}$ is relatively stable with respect to changes of $T$ for financial networks constructed as in Algorithm \ref{AlgorithmSection}. \\ 

In this section we plot different distance functions for a randomly chosen subset of the data to compare the impact on the MST and $\ORicci$. As theorized in \cite{StabilityOfMST}, similar distances did not present significant changes in their plots, except when considering  logarithms. Hence, we think it is reasonable to use $\sqrt{2(1-C_{i,j})}$ as it is commonly used in literature (where it is sometimes referred to as an  ``ultrametric'').  

\begin{figure}[H] % "[t!]" placement specifier just for this example
\centering
\textbf{Plots for different functions of $1-\rho_{i,j}$}
\begin{subfigure}{0.48\textwidth}
\includegraphics[width=\linewidth]{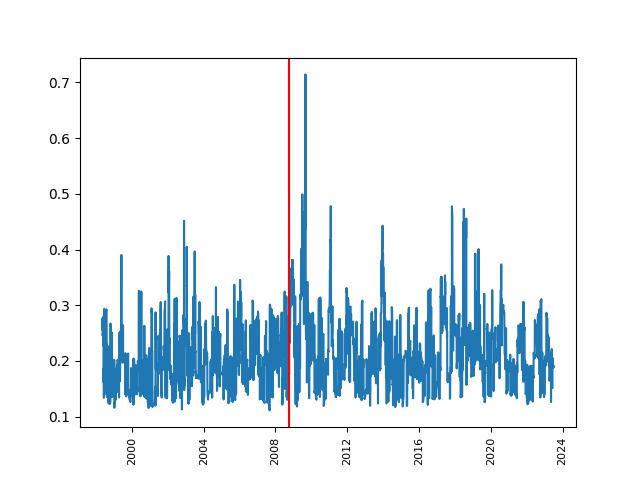}
\caption{$D_{i,j} = (2(1-\rho_{i,j}))^{1/4}$} \label{fig:acosts}
\end{subfigure}\hspace*{\fill}
\begin{subfigure}{0.48\textwidth}
\includegraphics[width=\linewidth]{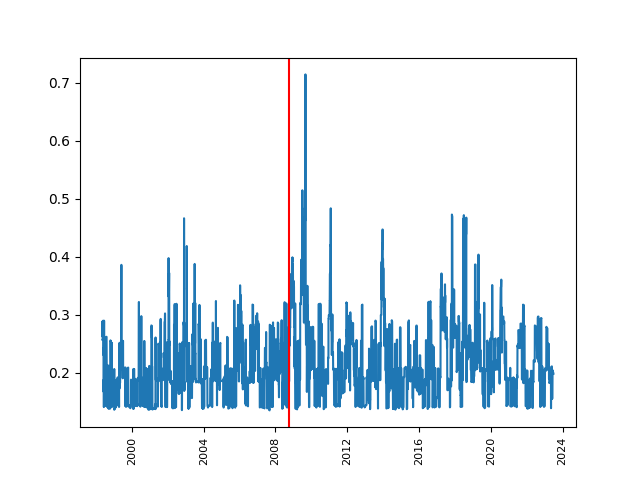}
\caption{$D_{i,j} = (2(1-\rho_{i,j}))^{1/16}$} \label{fig:bcosts}
\end{subfigure}

\medskip
\begin{subfigure}{0.48\textwidth}
\includegraphics[width=\linewidth]{10MockFourthRoot.png}
\caption{$D_{i,j} = (2(1-\rho_{i,j}))^{4}$} \label{fig:ccosts}
\end{subfigure}\hspace*{\fill}
\begin{subfigure}{0.48\textwidth}
\includegraphics[width=\linewidth]{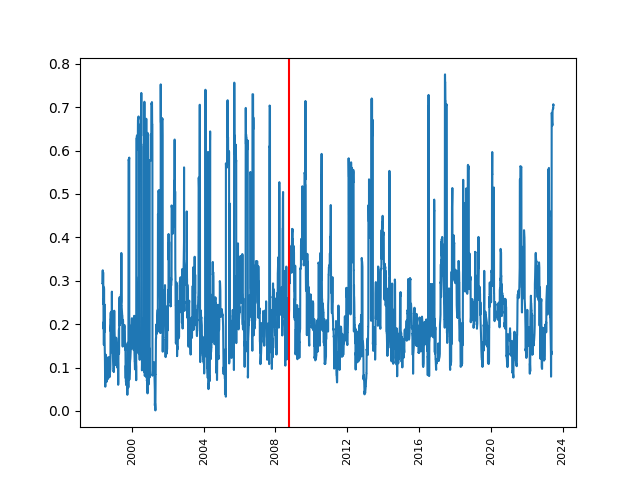}
\caption{ $D_{i,j} = \log(1 +(2(1-\rho_{i,j}))$} \label{fig:dcosts}
\end{subfigure}
\caption{Average $\ORicci$ curvature for different cost functions in the same (random) sub-graph of the 1997-2024 period with $T = 132$ and $\xi = 0.85$. The red vertical line corresponds to the crash-day of September 29 2008.}
\end{figure}

In Figures \ref{fig:acosts}, \ref{fig:bcosts}, \ref{fig:ccosts}, and \ref{fig:dcosts} we observe that for simple functions of the form $ D_{i,j} = (2(1-\rho_{i,j}))^p$ for some power $p$, the behaviour of the indicator can be similar. Nevertheless, a bad choice of cost function can result in noise, as exemplified by $\log(1+2(1-\rho))$. Notice that $h(a)= \log(1 + a)$ is concave in $[1,\infty)$ but Figure \ref{fig:dcosts} does not show the behaviour of the previous choices $h(a) = a^p$ (in Figures \ref{fig:acosts}, \ref{fig:bcosts}, and \ref{fig:ccosts}). The study of general properties of functions $h$ such that $h(1-\rho)$ is a reasonable indicator is an interesting future line of research.

\begin{example}(Weight independence) \label{WeightIndependence} \\
Consider the following network: 
\begin{figure}[h]
    \centering
    \includegraphics[scale = 0.4]{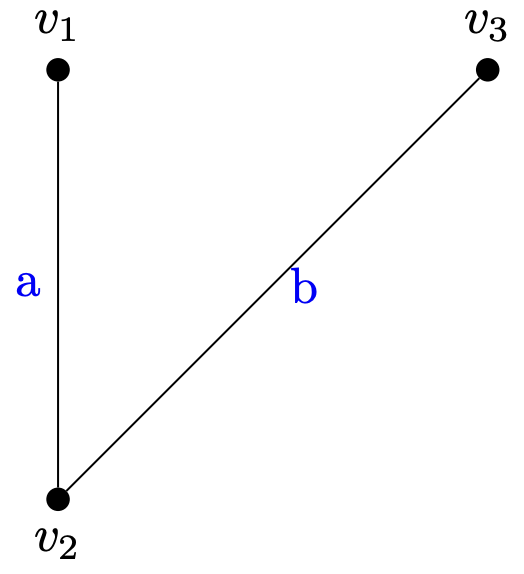}
    \caption{Independently of the weights $a$ and $b$ the associated $\ORicci$ curvature is always $\frac{1}{3}$.}
    \label{fig:AlwaysOneThird}
\end{figure}

We show that independently of the weights $a,b$, $\ORicci^{G} = 1/3.$ Note that $d(v_1,v_2) = 1, d(v_2,v_3) = 1, d(v_1,v_3) = 2$, and \begin{align}
    & \mu_{v_1} = \delta_{v_2}, \label{muv1always} \\
    & \mu_{v_2} = \frac{a}{a+b} \delta_{v_1} + \frac{b}{a+b} \delta_{v_3}, \\
    & \mu_{v_3} = \delta_{v_2}. \label{muv3always}
\end{align}
Consequently by \eqref{muv1always} and \eqref{muv3always}, $W^d(\mu_{v_1},\mu_{v_3}) = 2$ and 
\begin{align*}
    & W^{d}(\mu_{v_1},\mu_{v_2}) = 1 \Rightarrow k(v_1,v_2) = 0, \\
    & W^{d}(\mu_{v_2},\mu_{v_3})  = 1 \Rightarrow k(v_2,v_3) = 0.
\end{align*}
Hence, 
\begin{equation*}
  \text{Avg.} \ORicci^{G} = \frac{1}{3} \left(1 + 0 + 0 \right) = \frac{1}{3}.
\end{equation*}
\end{example}

Example \ref{WeightIndependence} shows that using the average Ricci curvature can be justified as a way of getting rid of the dependence of the weights used to obtain the MST. Of course, this independence is not the rule but a motivation for the use of the average. Below is another such example.

\begin{example}(Independence again) \label{independenceagain} \\
Consider the following network.
\begin{figure}[h]
    \centering
    \includegraphics[scale = 0.8]{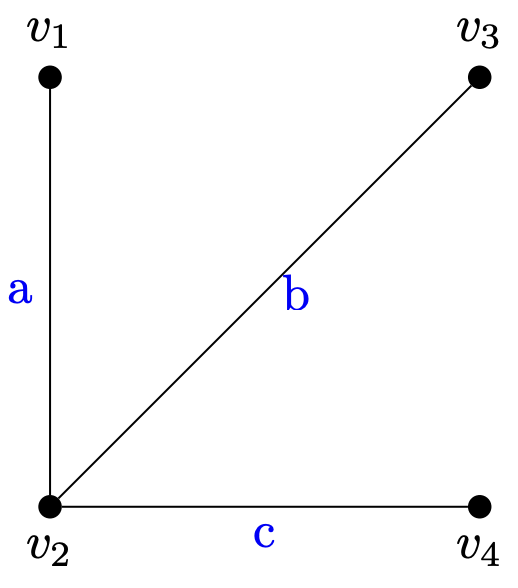}
    \caption{Independently of the weights $a,b,$ and $c$ the associated $\ORicci$ curvature is always $\frac{1}{3}$.}
    \label{fig:AlwaysIndep}
\end{figure}
The associated $\ORicci$ curvature is $1/3$, independently of the weights $a,b$ and $c$. The computations follow in the same manner as in Example \ref{WeightIndependence}, 
\begin{align*}
    & W^{d}(\mu_{v_1},\mu_{v_2}) = 1 \Rightarrow k(v_1,v_2) = 0, \\
    & W^{d}(\mu_{v_1},\mu_{v_3})  = 0 \Rightarrow k(v_2,v_3) = 1, \\
    & W^{d}(\mu_{v_1},\mu_{v_4})  = 0 \Rightarrow k(v_2,v_3) = 1, \\
    & W^{d}(\mu_{v_2},\mu_{v_3})  = 1 \Rightarrow k(v_2,v_3) = 0, \\
    & W^{d}(\mu_{v_2},\mu_{v_4})  = 1 \Rightarrow k(v_2,v_3) = 0, \\
    & W^{d}(\mu_{v_3},\mu_{v_4})  = 0 \Rightarrow k(v_2,v_3) = 1.
\end{align*}
Then,
\begin{equation}
    \ORicci^G = \frac{1}{6} \left( 0 + 1 + 1 + 0 + 0+ 1 \right) =  \frac{1}{3}.
\end{equation}
\end{example}

In Algorithm \ref{MainAlgorithm}, the weights $D_{i,j}$ are used to obtain the MST and to define the measures via \eqref{measuresFormulas}. An argument similar to those presented in \cite[Theorem 1]{WeightedTreeRicci} and \cite[Theorem 3]{WeightedTreeRicci} shows that the average Ricci curvature \textit{on a tree} is independent of the weights if we use the hop distance as the metric. This fact (except for the schematic construction of the tree) is an excellent motivation for the use of the \textit{average} $\ORicci$ curvature rather than a different statistic like the maximum or minimum Ricci curvature.

We conclude that the weight function $D_{i,j}$ should be chosen according to a prior knowledge of relative centralization of nodes. A promising line of investigation is to consider a Bayesian approach, adding prior information on node-centralization of a market to obtain better estimates for $\ORicci$.

\subsection{Adding high-value links} \label{XiSection}
In Algorithm \ref{MainAlgorithm}, we obtain the MST and then add back edges satisfying
\begin{equation} \label{ThresholdCondition}
     \rho_{i,j} \geq \xi.
\end{equation}
If $D_{i,j} = h(1-\rho_{i,j})$, where $h: \mathbb{R} \to \mathbb{R}$ is positive and continuous function with $h(0) = 0$, then \eqref{ThresholdCondition} adds the edges $ij$ whose corresponding weights $D_{i,j}$ are small. From this observation, we expect $\mu^*$ to be similar to $\mu$ due to the stability result of \cite{StabilityOfMST}. The exact nature of this similarity remains an open question. 

\section{Results} \label{sectionResults}
As in \cite{RicciFinance} we start by studying the American market. We compute the average $\ORicci$ curvature using historical closing prices for all companies in the S\&P 500, obtained via the \texttt{Yahoo Finance Python API}. The complete sample yields 388 companies for the 1997-2014 period.
\begin{figure}[H]
    \centering
    \textbf{$\ORicci^{Net}$ from 1998-2012.}\par
    \includegraphics[scale = 0.25]{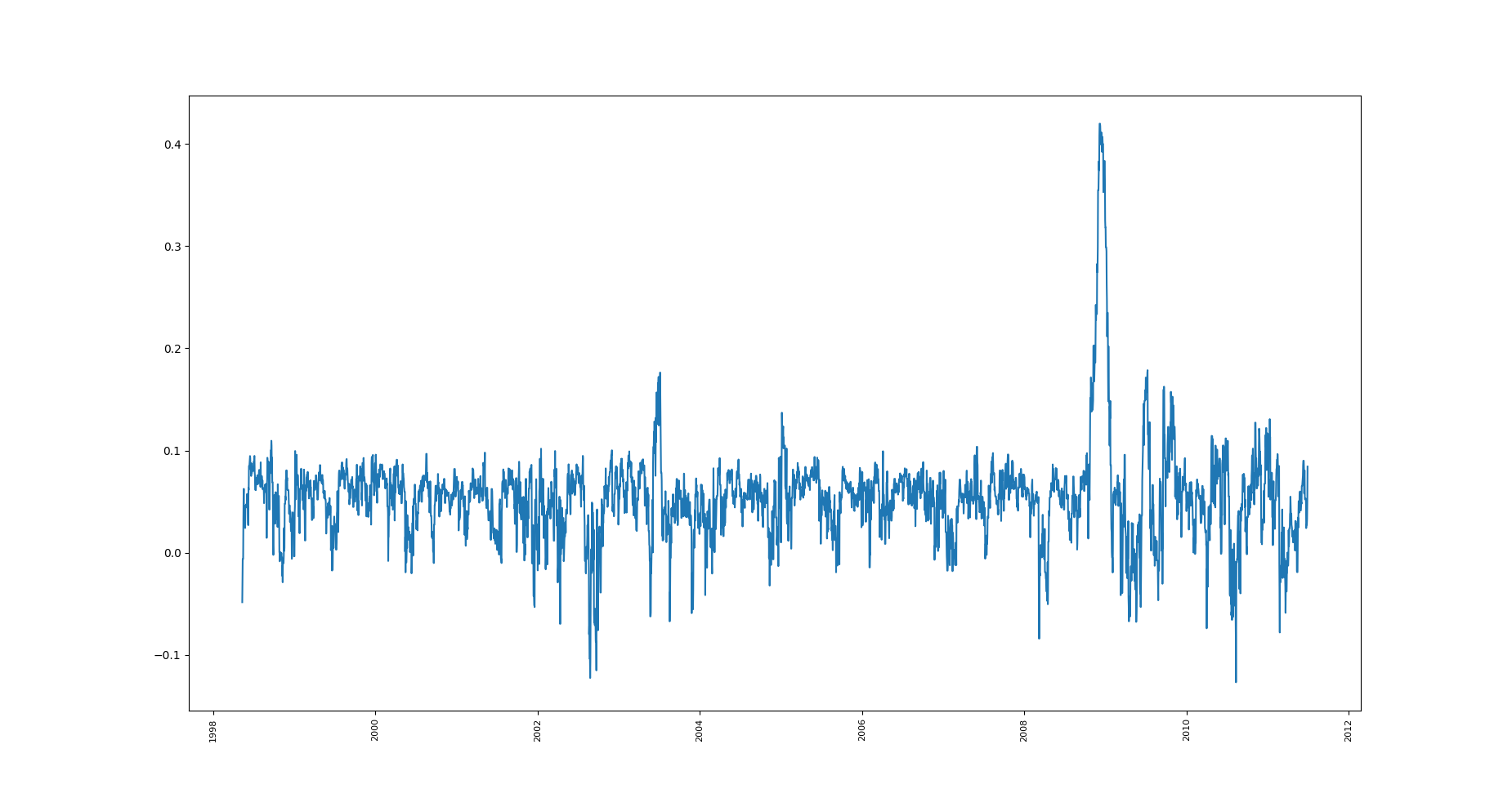}
    \caption{Value of the $\ORicci$ indicator in 1997-2012.}
    \label{fig:AllData2012}
\end{figure}

In Figure \ref{fig:AllData2012} we observe the value of the indicator from 1998-2012. We use all data available for S\&P listed companies, $T = 132$ and $\xi = 0.85$. We see a huge spike during September 2008, during the Lehman Brothers collapse.

In \cite{RicciFinance}, the algorithm of section \ref{MainAlgorithm} is presented as a measure of fragility in which the cases of most interest correspond to the biggest financial crisis (the 2008-2009 crisis being the most infamous). In Figure \ref{fig:AllData2012} we see the value of Algorithm \ref{MainAlgorithm} peak at the same time as the financial crisis of 2008-2009. In the next figure, we observe the behaviour during the 2008-2009 year. 

\begin{figure}[H]
    \centering
    \textbf{$\ORicci^{Net}$ during the financial crisis of 2008-2009 for $T = 132$ and $\xi= 0.85$.}\par
    \includegraphics[scale = 0.25]{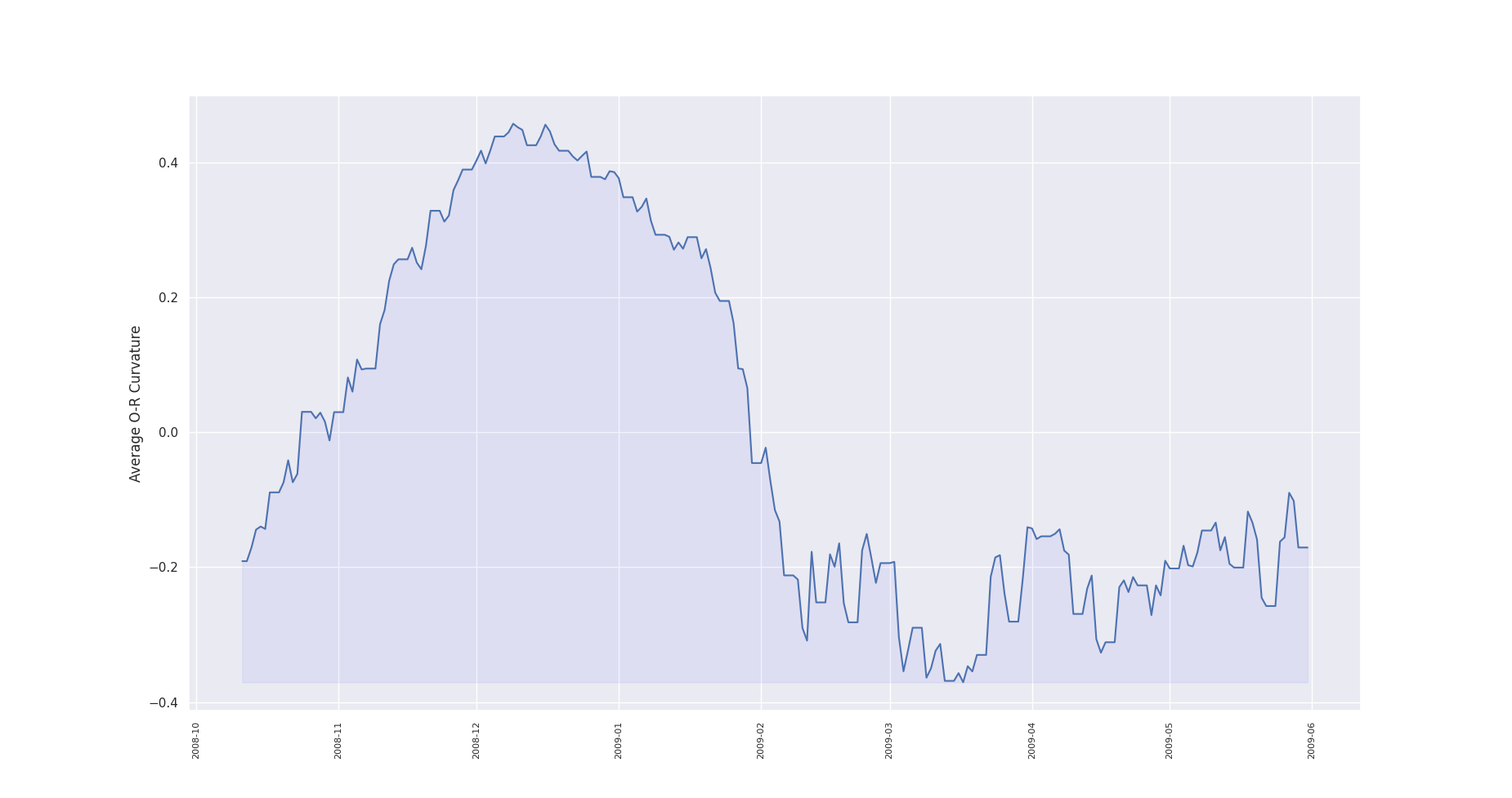}
    \caption{Value of the $\ORicci$ indicator in 2008-2009 (during the financial crisis) using all the companies in the S\&P 500, $T = 132$, and $\xi = 0.85$.}
    \label{fig:AllData20082009Xi132}
\end{figure}

Figure \ref{fig:AllData20082009Xi132} shows that the $\ORicci$ seems to have returned to it's usual values by February or March 2009. Nevertheless, the instability of the market seems to be different when we analyze with smaller values of $T$.

%We can study not only relative heights to compare financial crashes but other underlying features. 
%If one aims to compare two periods of financial instability, in a more quantitative way, it is not enough to only compare the heights of $\ORicci^{net}$. Comparing two periods of financial instability can be done through the areas under the curves of $\ORicci^{NEt}$.\\

\begin{figure}[H]
    \centering
    \textbf{$\ORicci^{Net}$ during the financial crisis of 2008-2009 for $T = 22$ and $\xi= 0.85$.}\par
    \includegraphics[scale = 0.25]{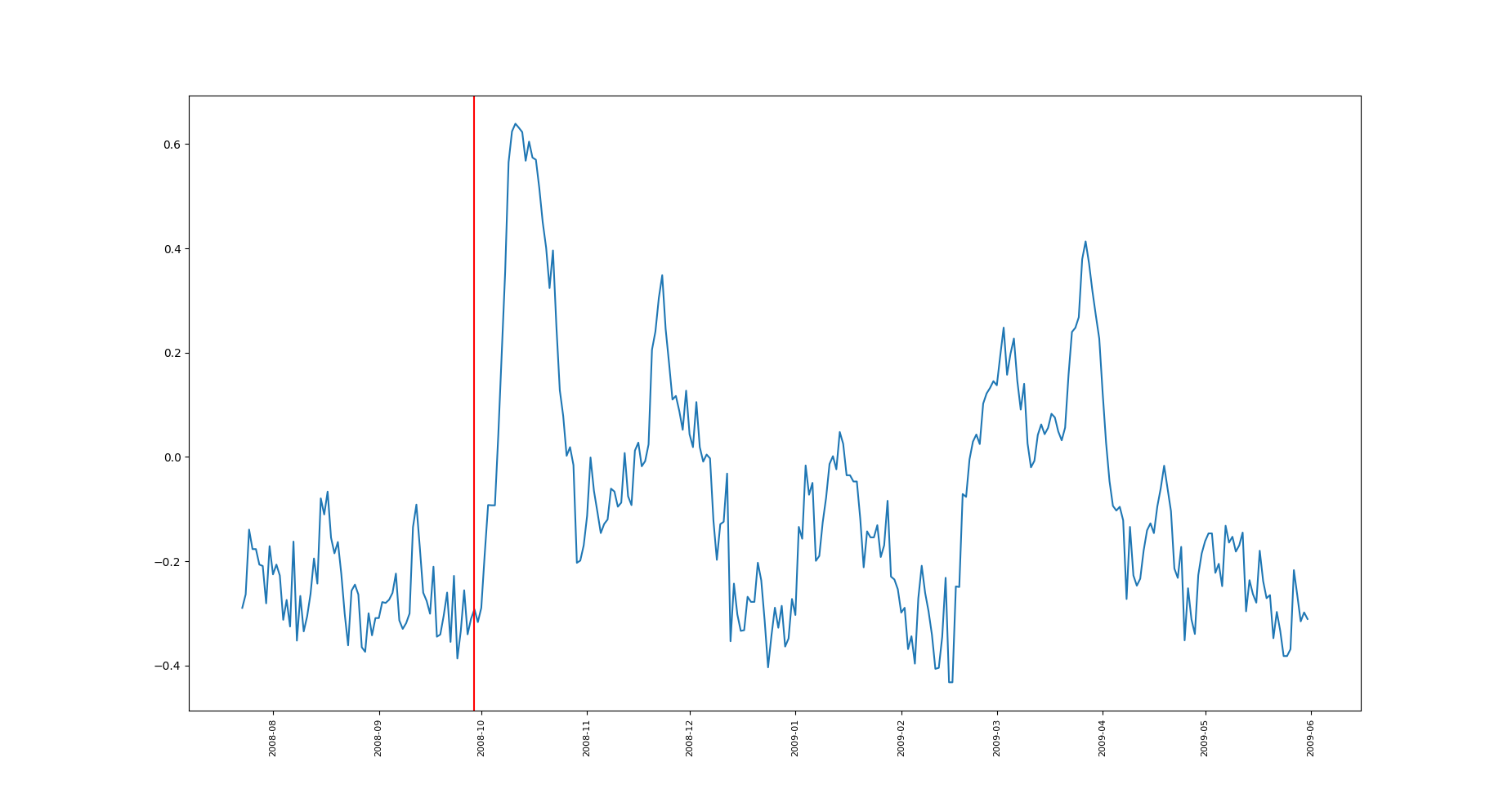}
    \caption{Value of the $\ORicci$ indicator in 2008-2009 (during the financial crisis) using all the companies in the S\&P 500, $T = 22$, and $\xi = 0.85$. In red we observe the crash day September 29, 2008.}
    \label{fig:AllData20082009}
\end{figure}

In Figure \ref{fig:AllData20082009}, we see the analysis of the year 2008-2009. The 2008 crisis corresponds to the period of instability, the variance and height of the average $\ORicci$ curvature did not return to it's pre-crisis values until April-May 2009.

\subsection{American market: Change of threshold}
In the works \cite{MarketMining} and \cite{NetworkPerspective} it was argued that adding high-value links tends to give a more accurate schematic picture of the actual state of the market. For the method of Algorithm \ref{MainAlgorithm} to be robust, we need to analyze the $\xi$-elasticity of the $\ORicci$ indicator. By Proposition \ref{UpperBoundCurvature} we can estimate the effect of adding a single edge. In particular, the edges associated to high-value links can only decrease the hop distance in the new graph (because any path of less steps remains in the graph). 

In the following figures we observe that as we increase $\xi$, the average curvature decreases. Although we see this behaviour on the general level of the plots, it is not true on a ``point-to-point'' basis. This observation, rooted in the fact that $\ORicci$ depends on the ratio of $W^{d^*}$ and $d^*$, indicates that an apparent trend in a small period of time is not reliable unless $\xi$ is somehow prescribed or fixed. Two experiments with similar values $\xi_1 \approx \xi_2 $ and $ \xi_1 \neq \xi_2$ can indicate short-term tendencies in opposite directions.

%%%%%%%%%%%%%%%%%%%%%%%%%%%%%%%%%%%%%%%%%%%%%%%%%%%%%%%%% T = 22
We start with an extreme value of $T= 2000$ and then we show that the same behaviour is observed in the cases $T = 22$ and $T=132$, which are much noisier by default.
%%%%%%%%%%%%%%%%%%%%%%%%%%%%%%%%%%%%%%%%%%%%%%%%%%%%%%%%%%%%%%% T = 2000

\begin{figure}[H]
    \centering
    \includegraphics[scale = 0.25]{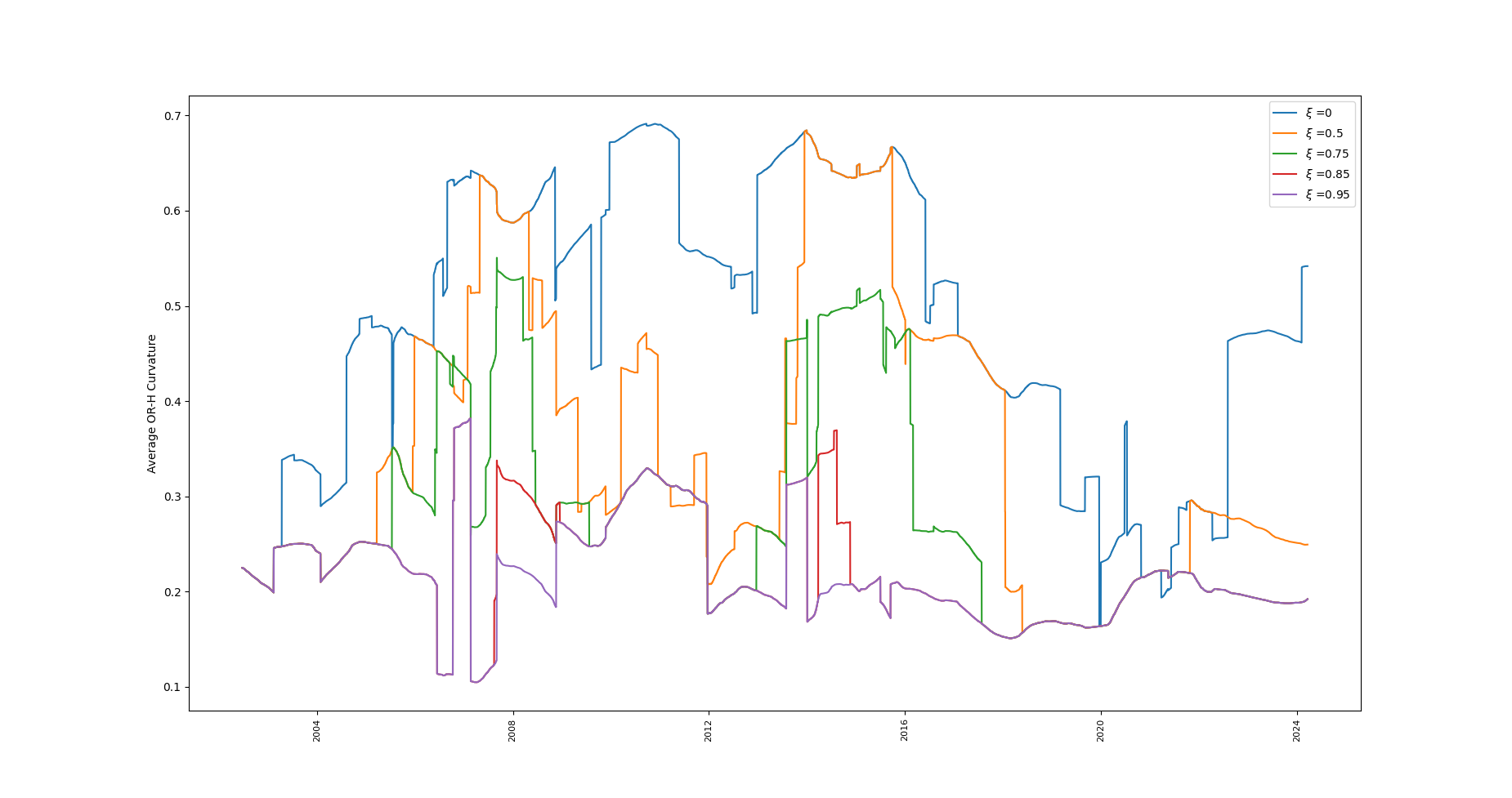}
    \caption{Changes on the method as $\xi$ varies, full time period 1997-2024, on a random sub-graph of 10 companies of the American market.}
    \label{fig:Xi1}
\end{figure}
\begin{figure}[H]
    \centering
    \includegraphics[scale = 0.25]{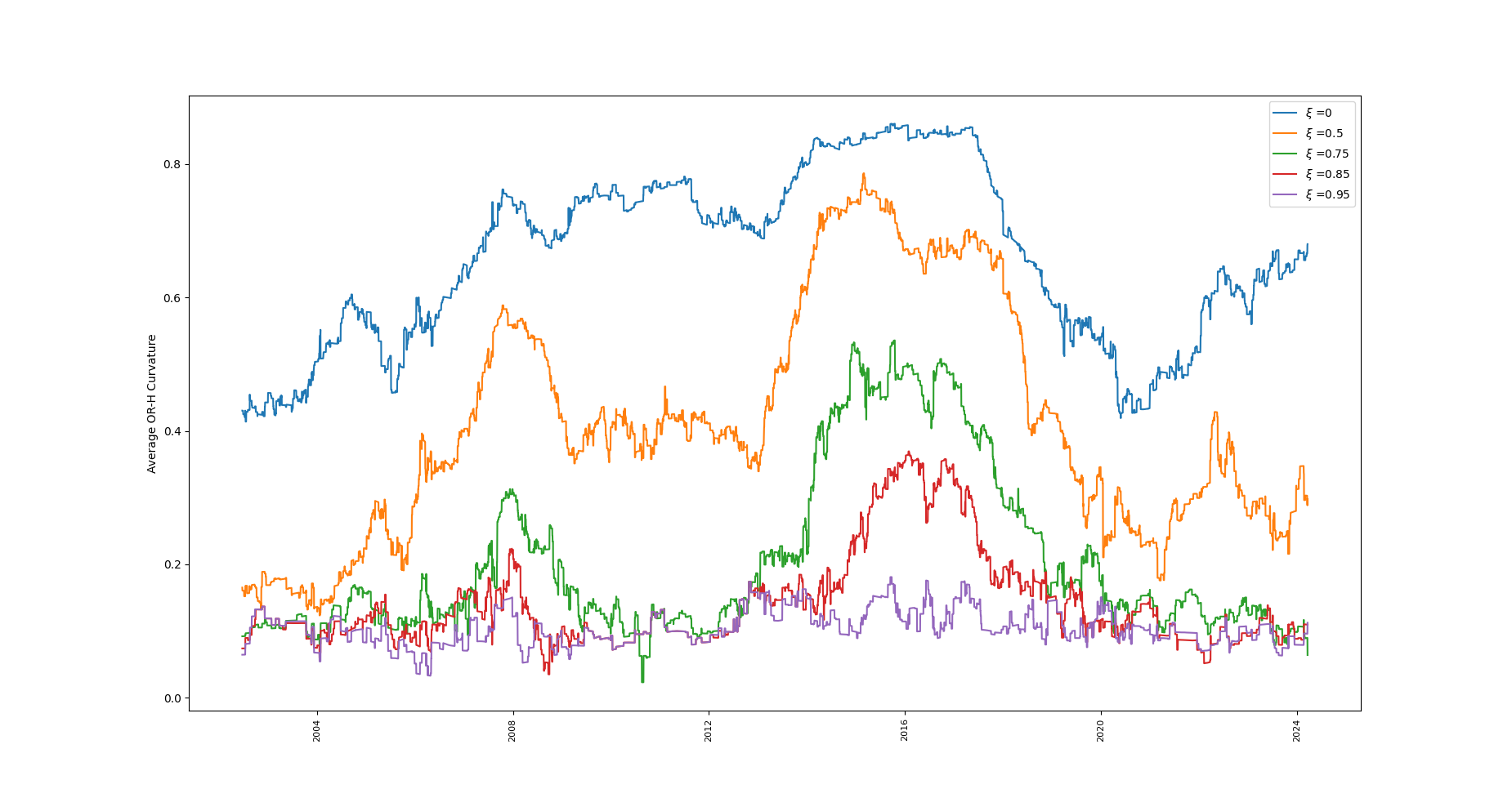}
    \caption{Changes on the method as $\xi$ varies, full time period 1997-2024, on a random sub-graph of 25 companies of the American market.}
    \label{fig:Xi3}
\end{figure}
\begin{figure}[H]
    \centering
    \includegraphics[scale = 0.25]{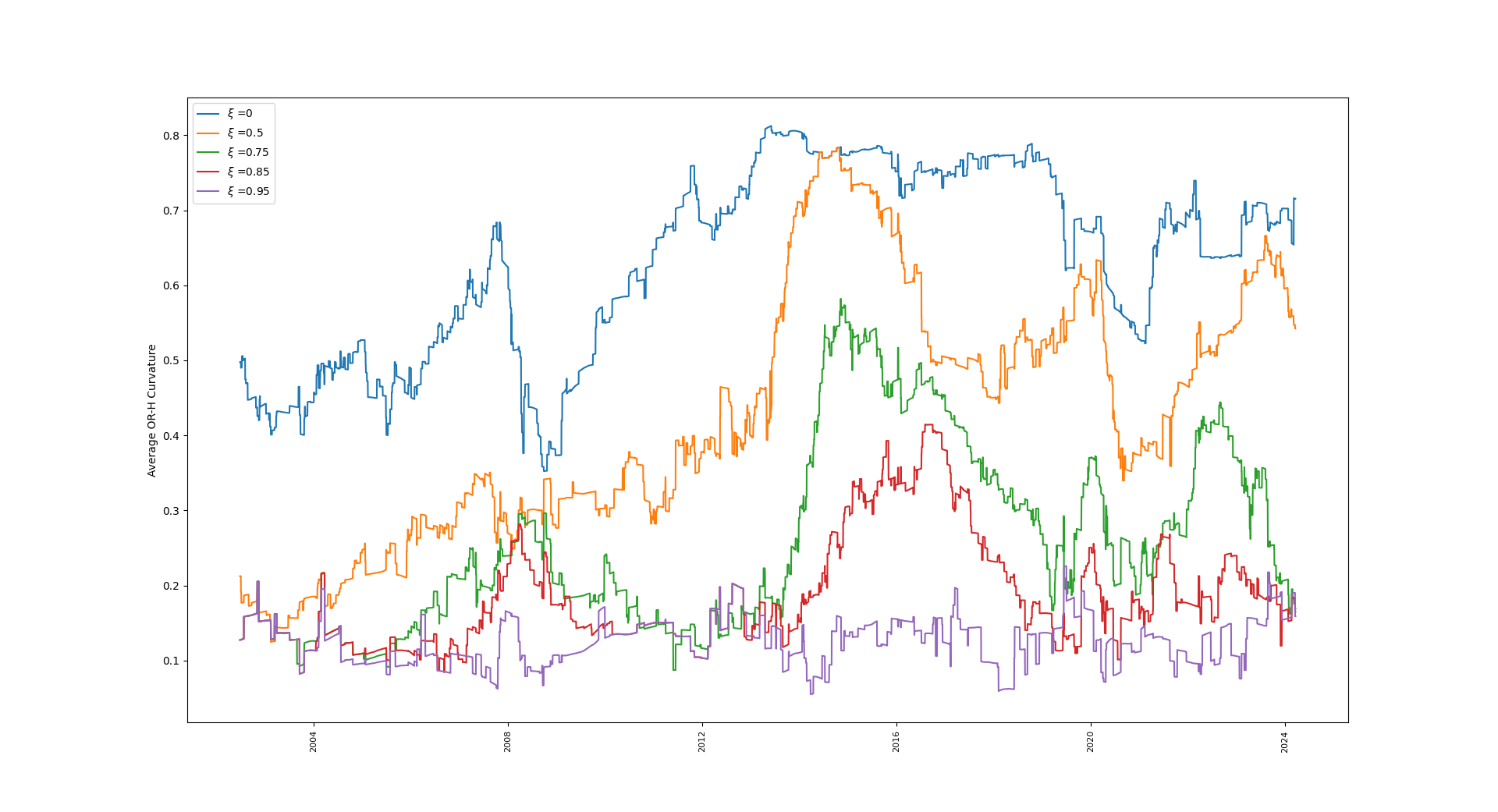}
    \caption{Changes on the method as $\xi$ varies, full time period 1997-2024, on a random sub-graph of 30 companies of the American market.}
    \label{fig:Xi2}
\end{figure}
\begin{figure}[H]
    \centering
    \includegraphics[scale = 0.25]{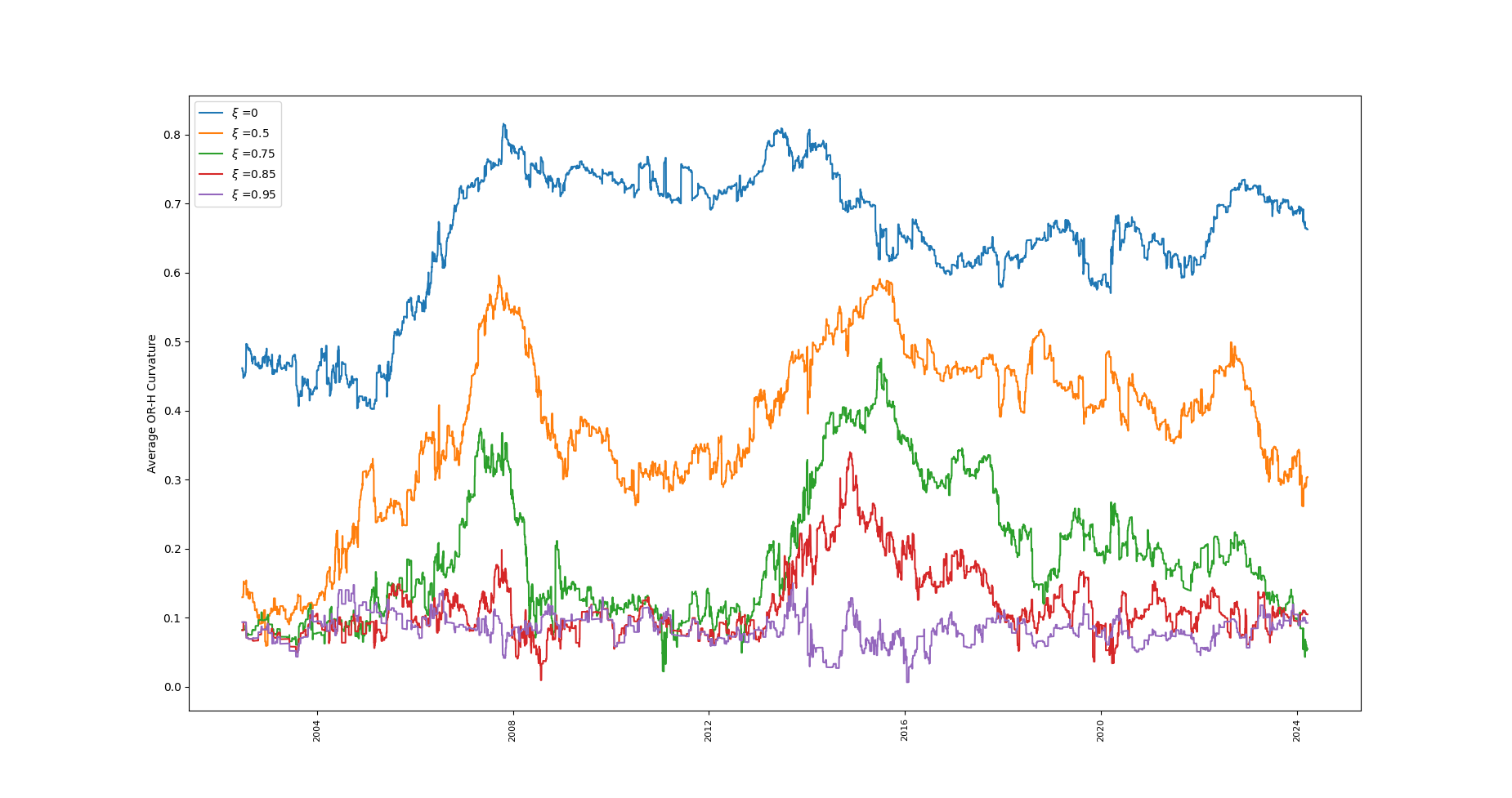}
    \caption{Changes on the method as $\xi$ varies, full time period 1997-2024, on a random sub-graph of 50 companies of the American market.}
    \label{fig:Xi4}
\end{figure}

Figures \ref{fig:Xi1}, \ref{fig:Xi2}, \ref{fig:Xi3} and \ref{fig:Xi4} show that although the general trend is captured by the indicator, changes in $\xi$ are extremely significant for short-period analysis.

Next, we show the same behaviour is observed in the noisy case (arguably the case of most interest) corresponding to smaller values of $T$.

\begin{figure}[H]
    \centering
    \includegraphics[scale = 0.25]{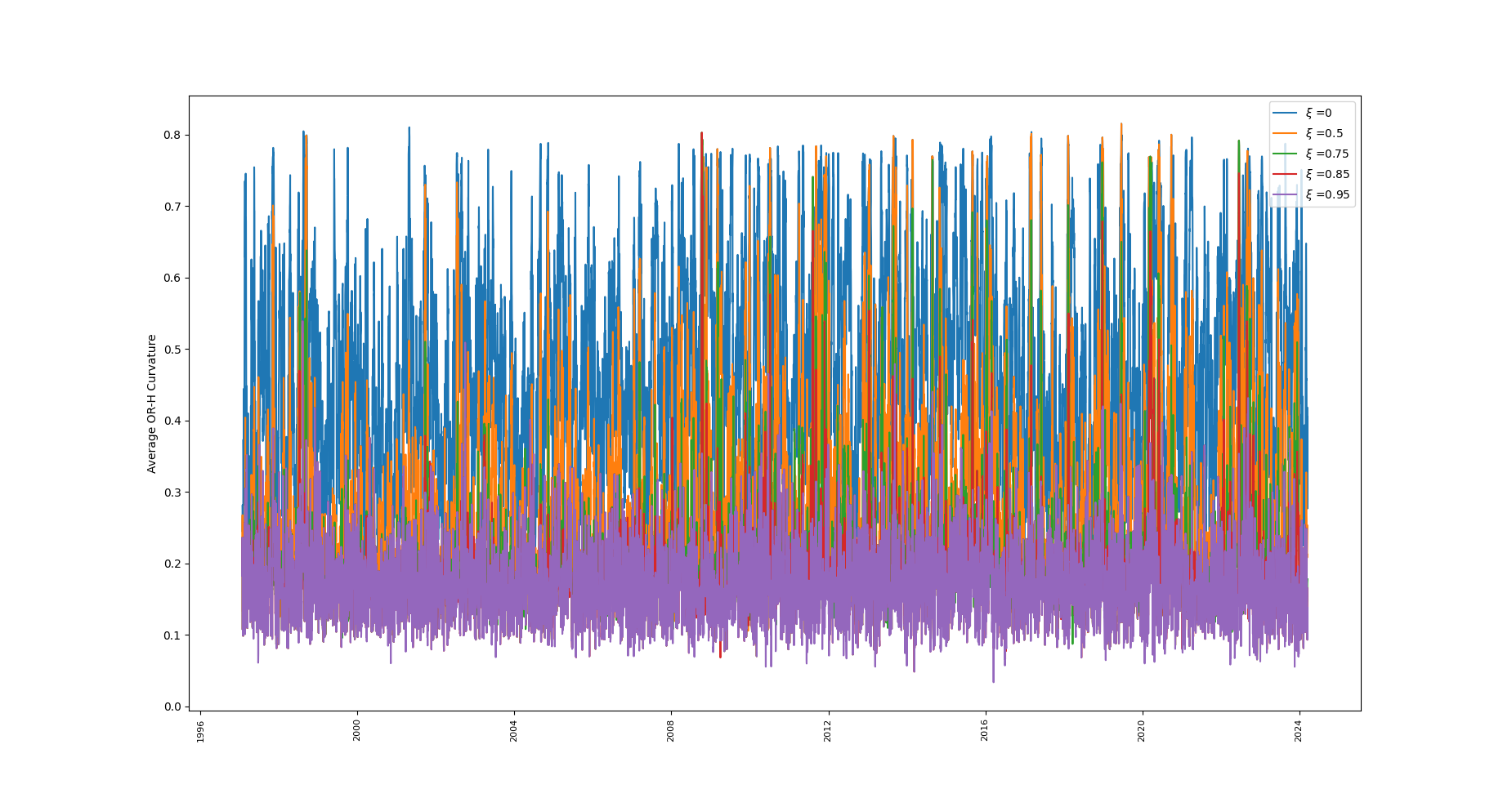}
    \caption{Changes on the method as $\xi$ varies, full time period 1997-2024, on a random sub-graph of 10 companies of the American market.}
    \label{fig:Xi1T22}
\end{figure}
\begin{figure}[H]
    \centering
    %%%%%%%%%%%%%%%%%%%%%%%%%%%%%% NEEDS CHANGE
    \includegraphics[scale = 0.25]{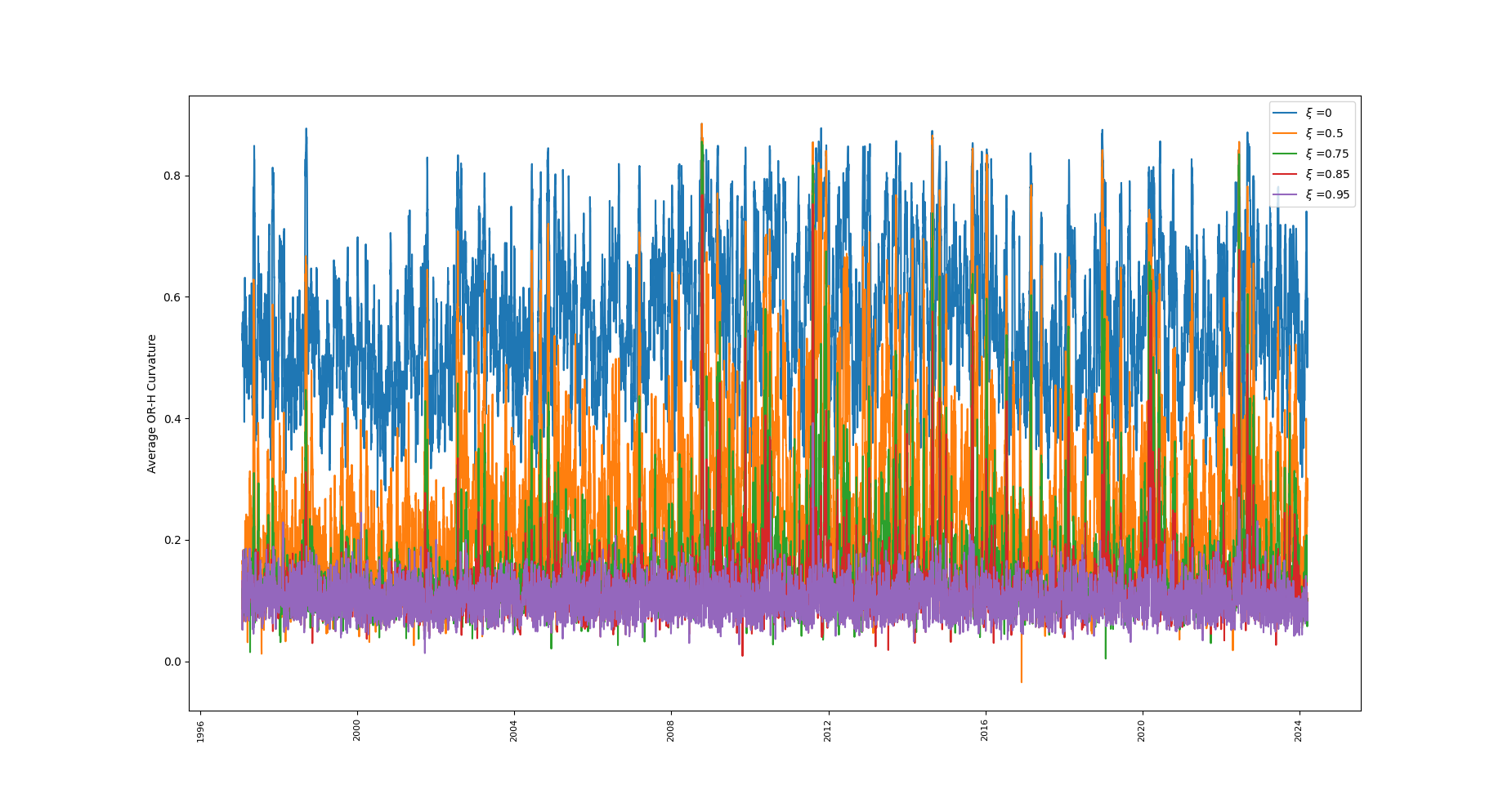}
    \caption{Changes on the method as $\xi$ varies, full time period 1997-2024, on a random sub-graph of 30 companies of the American market.}
    \label{fig:Xi2T22}
\end{figure}
%\begin{figure}[H]
%    \centering
%    %%%%%%%%%%%%%%%%%%%%%%%%%%%%%%%% NEEDS CHANGE
%    \includegraphics[scale = 0.25]{Changes_Xi_25.png}
%    \caption{Changes on the method as $\xi$ varies, full time period 1997-2024 on a random subgraph of 25 companies of the american market.}
%    \label{fig:Xi3T22}
%\end{figure}
\begin{figure}[H]
    \centering
    %%%%%%%%%%%%%%%%%%%%%%%%%%%%%%%%%% NEEDS CHANGE
    \includegraphics[scale = 0.25]{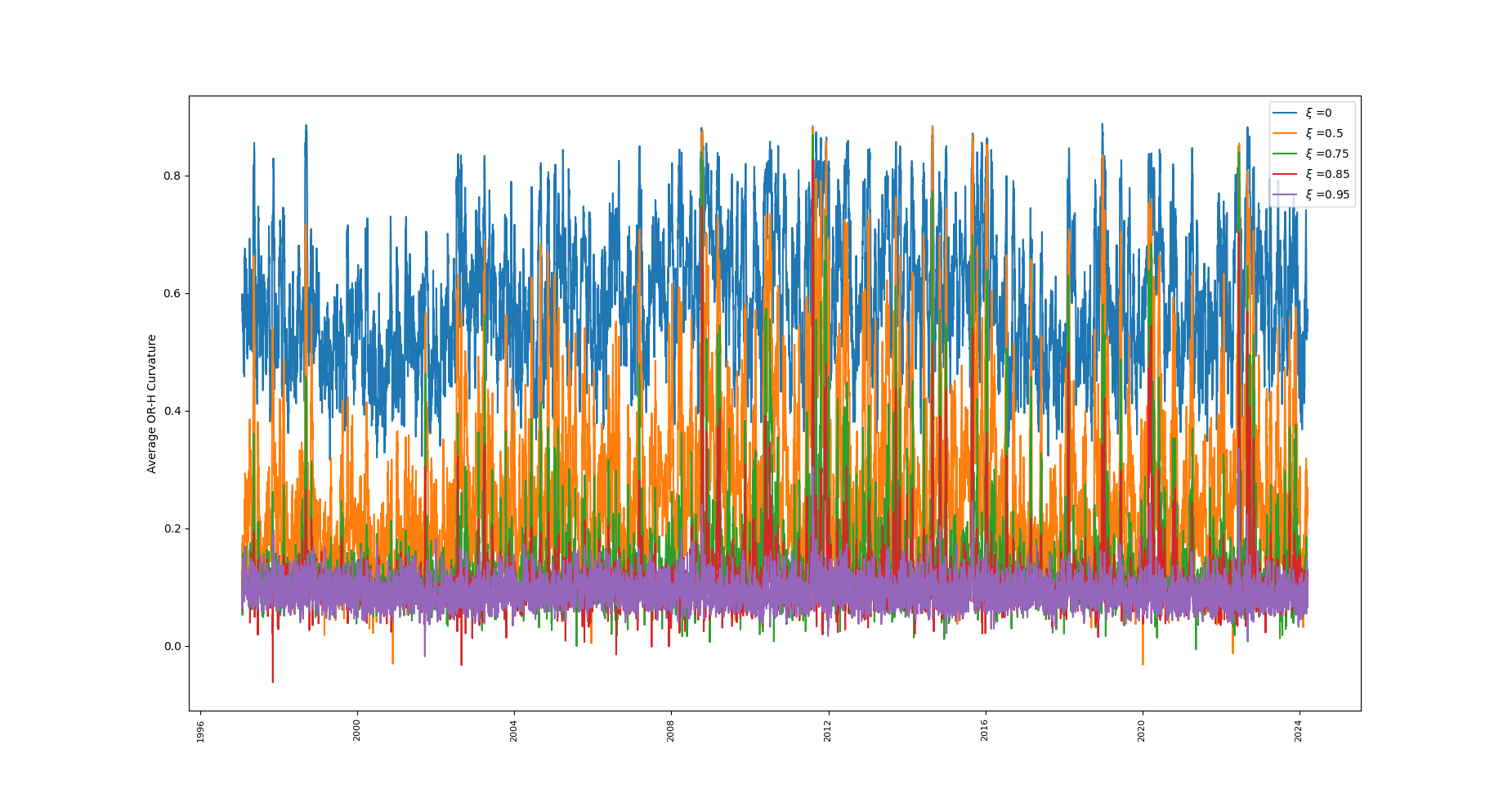}
    \caption{Changes on the method as $\xi$ varies, full time period 1997-2024, on a random sub-graph of 50 companies of the American market.}
    \label{fig:Xi4T22}
\end{figure}

We continue the same simulations for $T = 132$ in order to show that the phenomena explained in the previous section prevails in a slightly smoother case.

\begin{figure}[H]
    \centering
    %%%%%%%%%%%%%%%%%%%%%%%%%%%%%%%%%% NEEDS CHANGE
    \includegraphics[scale = 0.25]{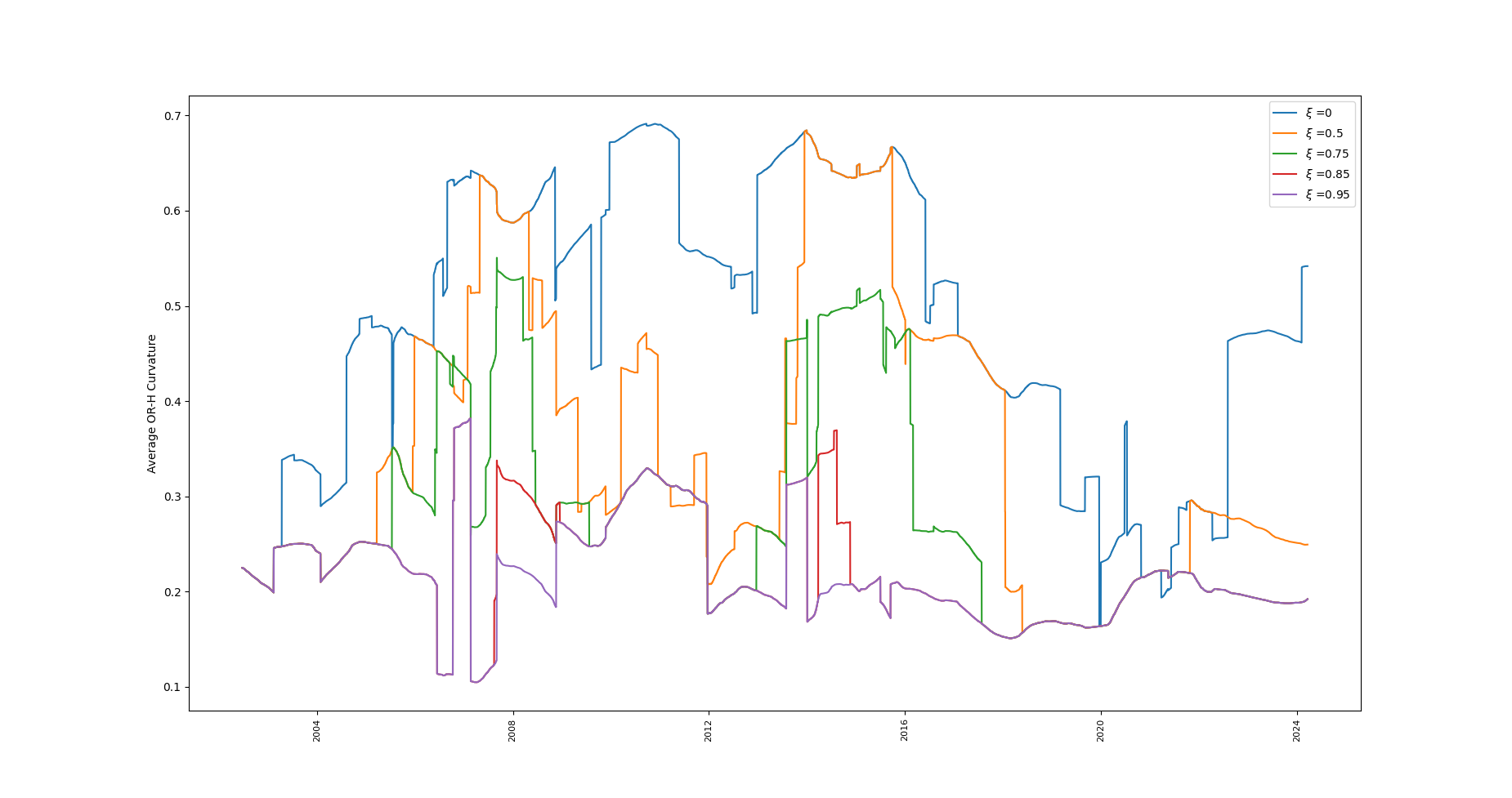}
    \caption{Changes on the method as $\xi$ varies, full time period 1997-2024 on a random sub-graph of 10 companies of the American market, $ T = 22$.}
    \label{fig:Xi1T132}
\end{figure}
\begin{figure}[H]
    \centering
    %%%%%%%%%%%%%%%%%%%%%%%%%%%%%%%%%% NEEDS CHANGE
    \includegraphics[scale = 0.25]{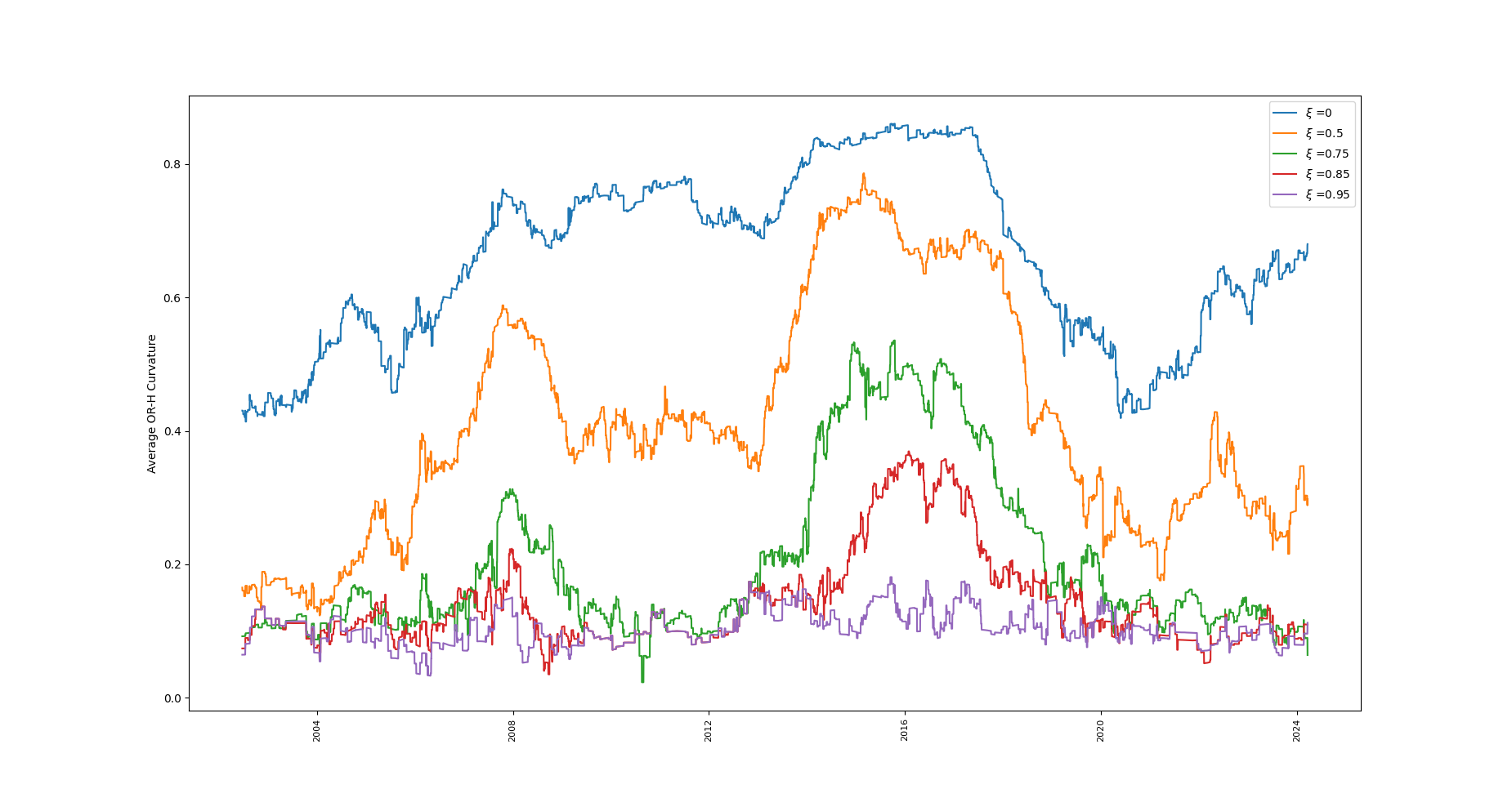}
    \caption{Changes on the method as $\xi$ varies, full time period 1997-2024 on a random sub-graph of 25 companies of the American market, $T = 22$.}
    \label{fig:Xi3T132}
\end{figure}
\begin{figure}[H]
    \centering
    %%%%%%%%%%%%%%%%%%%%%%%%%%%%%%%%%% NEEDS CHANGE
    \includegraphics[scale = 0.25]{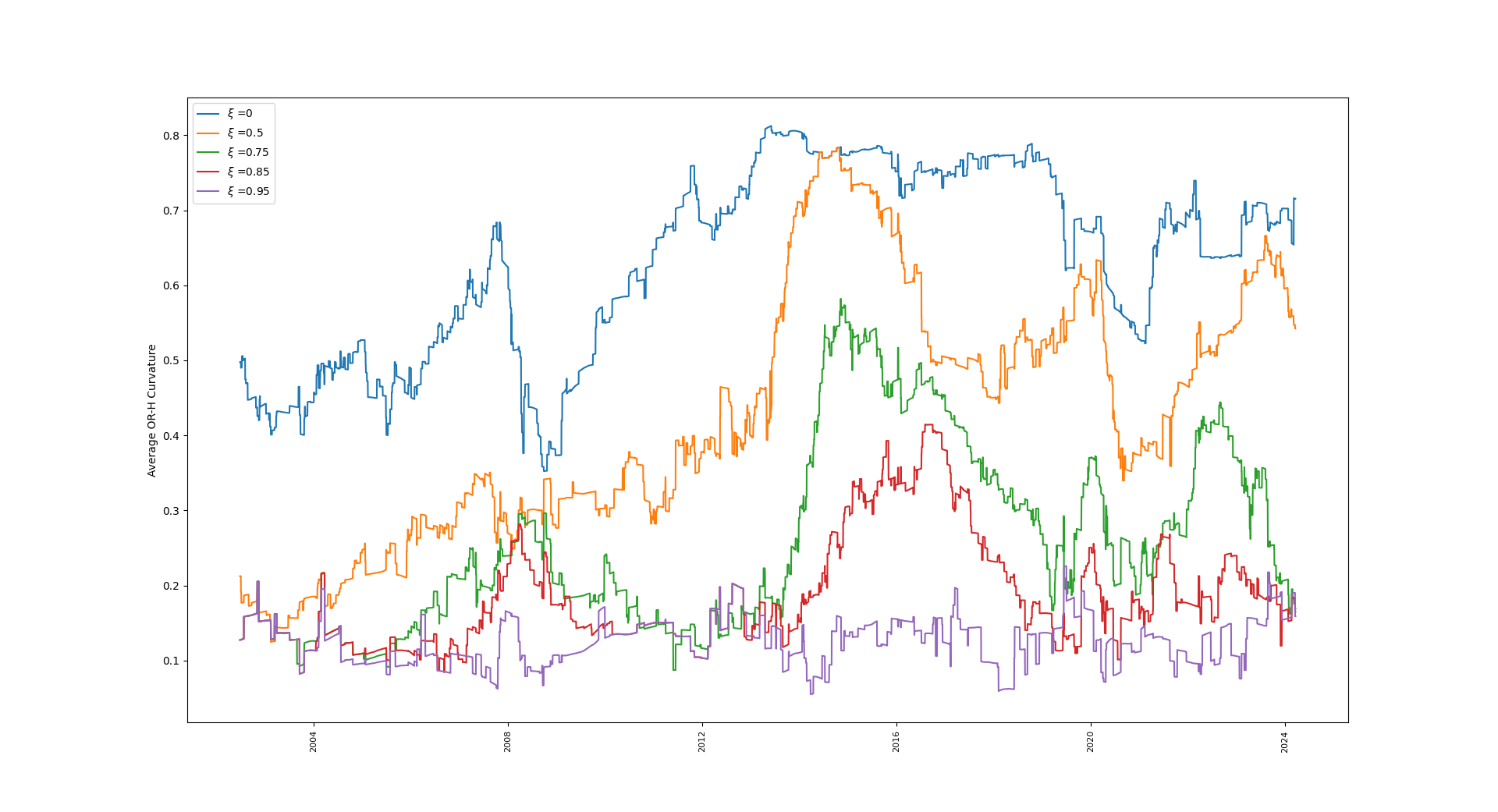}
    \caption{Changes on the method as $\xi$ varies, full time period 1997-2024 on a random sub-graph of 30 companies of the American market, $T = 22$.}
    \label{fig:Xi2T132}
\end{figure}
\begin{figure}[H]
    \centering
    %%%%%%%%%%%%%%%%%%%%%%%%%%%%%%%%%% NEEDS CHANGE
    \includegraphics[scale = 0.25]{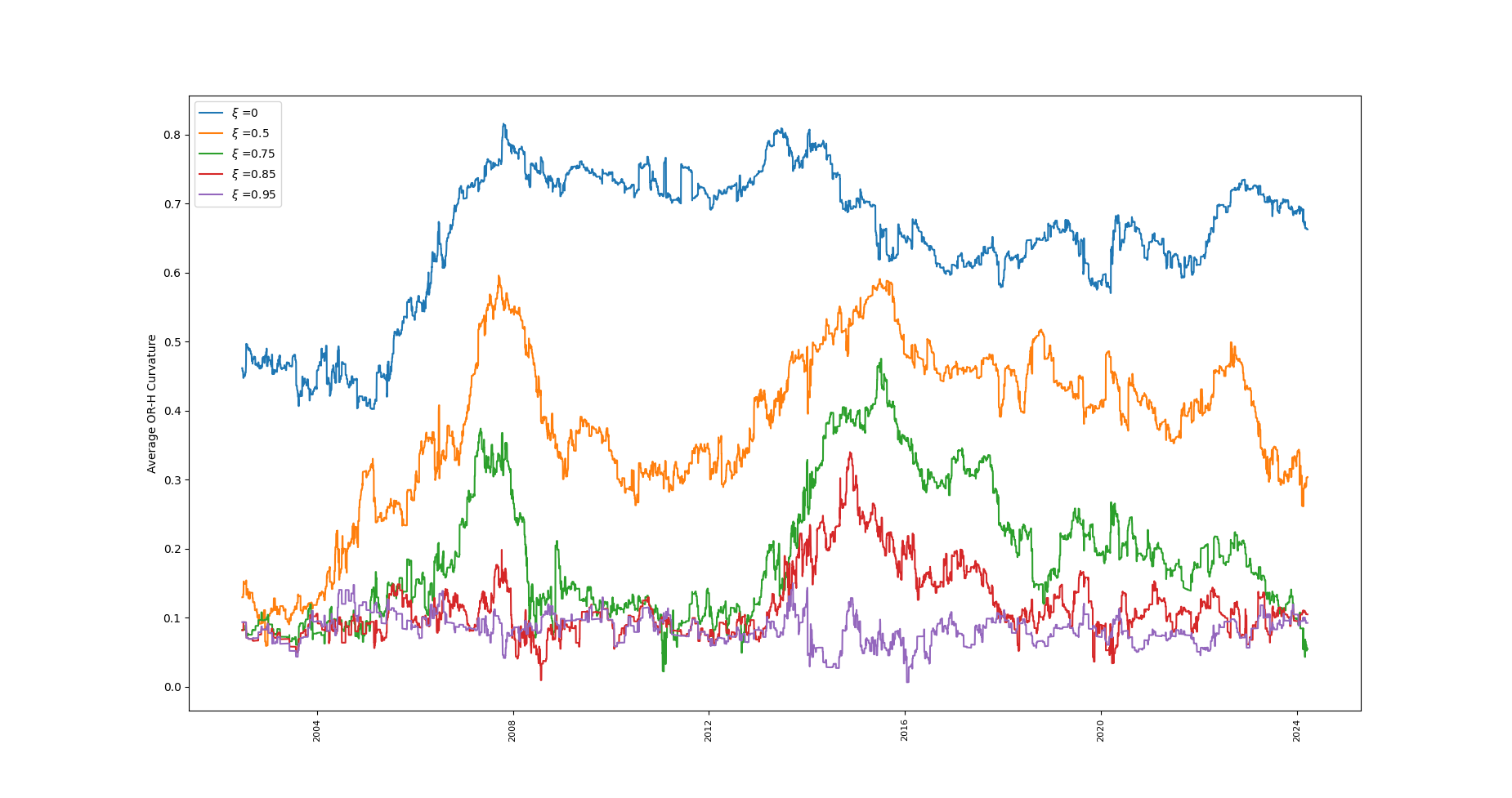}
    \caption{Changes on the method as $\xi$ varies, full time period 1997-2024 on a random sub-graph of 50 companies of the American market, $T = 22$.}
    \label{fig:Xi4T132}
\end{figure}

\subsection{Different market: The Canadian experience} \label{CountriesSection}
In this section we choose companies from the S\&P ETFs for Canadian companies and evaluate the $\ORicci$ estimator.

\begin{figure}[H]
    \centering
    \textbf{Canadian experience 2008-2024}\par
    \includegraphics[scale = 0.5]{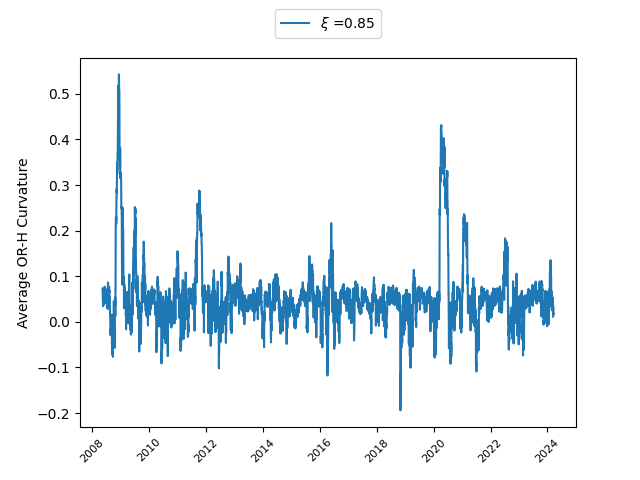}
    \caption{$\ORicci^{Net}$ during the period from 2008-2024 for all companies domiciled to Canada according to TMX.}
    \label{fig:TMX}
\end{figure}

Figure \ref{fig:TMX} shows the relative difference (according to $\ORicci^{net}$) of the financial crisis of 2008 and the impact of the Covid-19 pandemic on $\ORicci^{net}$. From \ref{fig:TMX} we see that the impact of the 2008 crisis on correlations of the stock market is higher than that of the pandemic. In the following figure (Figure \ref{fig:TMXXi}) we see this difference replicated for most values of $\xi$, making the analysis much more robust.

\begin{figure}[H]
    \centering
    \textbf{Canadian experience 2008-2024 changing $\xi$}\par
    \includegraphics[scale = 0.5]{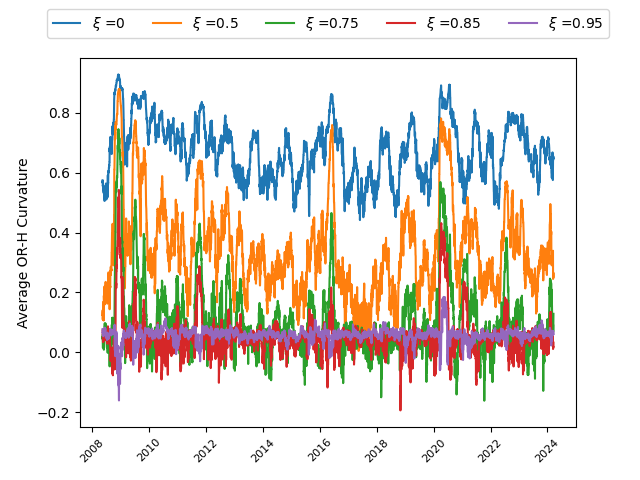}
    \caption{$\ORicci^{Net}$ during the period from 2008-2024 for all companies domiciled to Canada according to TMX, for different values of $\xi$.}
    \label{fig:TMXXi}
\end{figure}
\begin{figure}[H] % "[t!]" placement specifier just for this example
\centering
\textbf{Average $\ORicci$ Curvature of the top 20 Canadian companies from 1997-2024.}\par

\begin{subfigure}{0.48\textwidth}
\includegraphics[width=\linewidth]{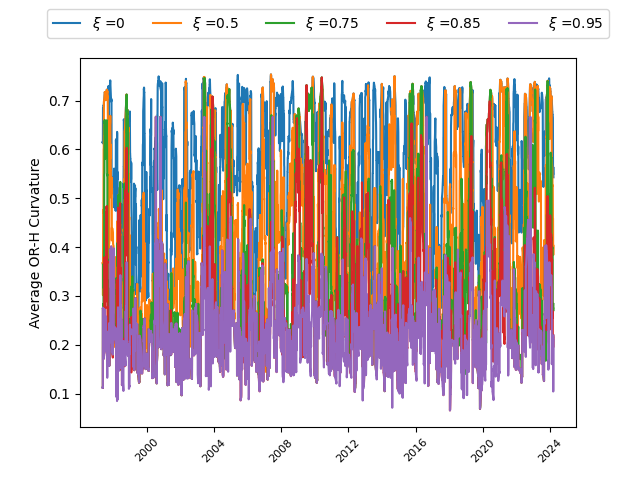}
\caption{$T = 132$} \label{fig:ACAN}
\end{subfigure}\hspace*{\fill}
\begin{subfigure}{0.48\textwidth}
\includegraphics[width=\linewidth]{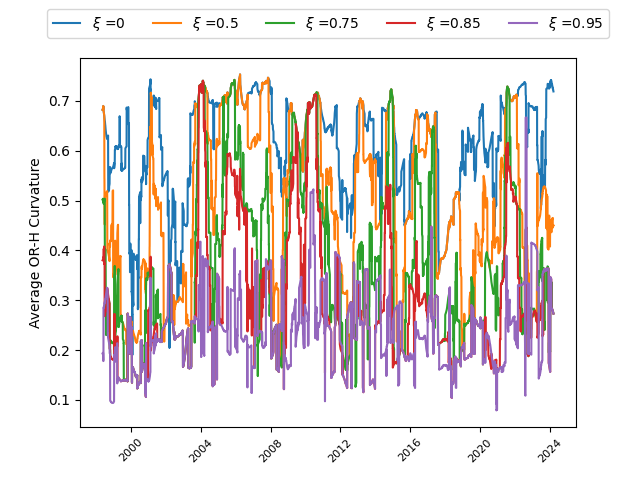}
\caption{$T = 500$} \label{fig:BCAN}
\end{subfigure}

\medskip
\begin{subfigure}{0.48\textwidth}
\includegraphics[width=\linewidth]{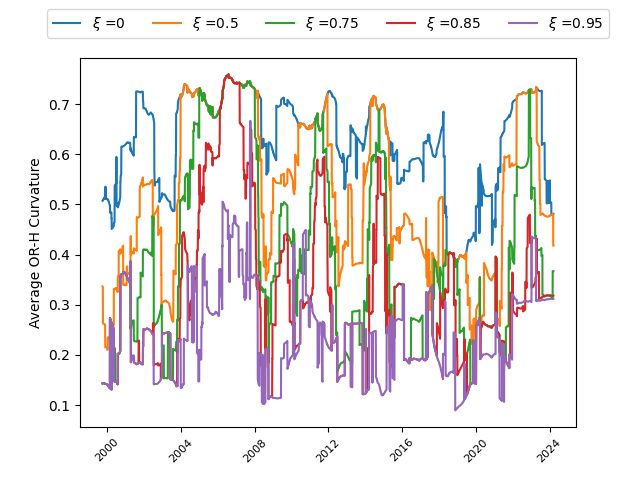}
\caption{$T = 1000$} \label{fig:CCAN}
\end{subfigure}\hspace*{\fill}
\begin{subfigure}{0.48\textwidth}
\includegraphics[width=\linewidth]{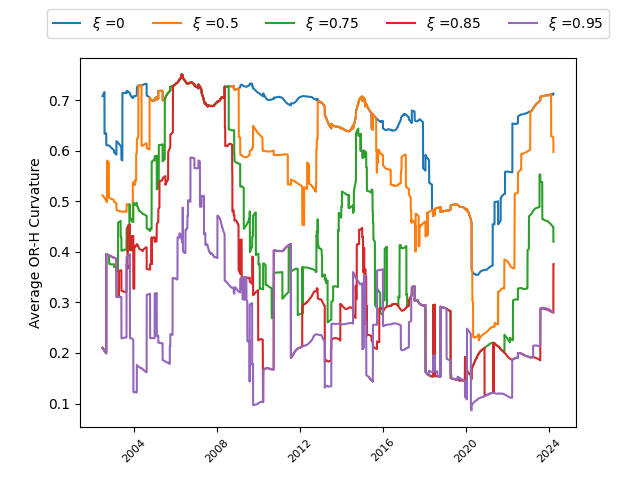}
\caption{ $T  = 2000$} \label{fig:DCAN}
\end{subfigure}
\caption{Average $\ORicci$ curvature (Algorithm \ref{MainAlgorithm}) for the top 20 Canadian companies composing the S\&P BMI Canada for different values of $\xi$.}
\end{figure}

From the sub-graph of the Canadian market with large $T$ (Figures \ref{fig:CCAN} and \ref{fig:DCAN}), we observe a clear increase in fragility during 2004-2008 period and a clear decrease in the indicator for all values of $\xi$ in 2018-2020, followed by another increase at the same time as the Covid-19 pandemic.

\subsection{Not a country: The Tech-sector crunch} \label{SectorsSection}
In this section we look at the $\ORicci$ curvature of a specific subgraph of the market. We look at the technology sector in order to analyze it's fragility before the famous Tech crunch resulting in the lay-offs of 2023-2024. As before, for values of $\xi \in (0.75,0.85)$ the plot seems more stable. Nevertheless, the indicator does not show the fragility of the sector in a predictive manner (one possible explanation for the decrease in curvature are the lay-offs themselves). We use the time period 2018-2024 instead of 1997-2024 as most current tech companies have had their IPOs in more recent times.
\begin{figure}[H]
    \centering
    \textbf{Technology Sector} \\
    \includegraphics[scale = 0.5 ]{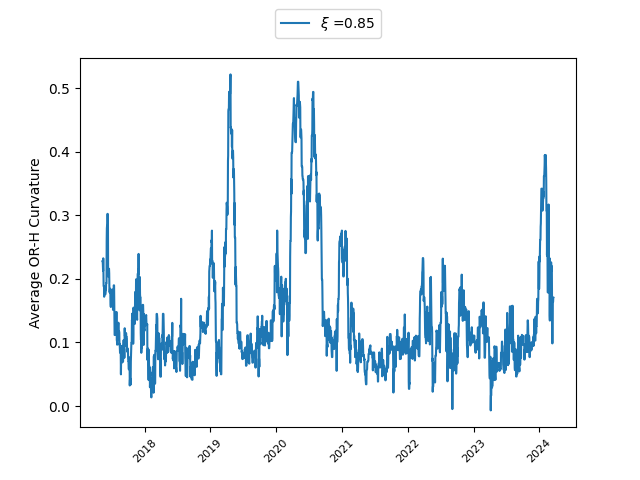}
    \caption{Average $\ORicci$ curvature, time period 2018-2024, all the companies listed in NASDAQ Tech Sector (NDXT), $T = 132$, and $\xi = 0.85$.}
    \label{fig:SectorT132}
\end{figure}
\begin{figure}[H]
    \centering
    \textbf{Technology Sector} \\
    \includegraphics[scale = 0.5 ]{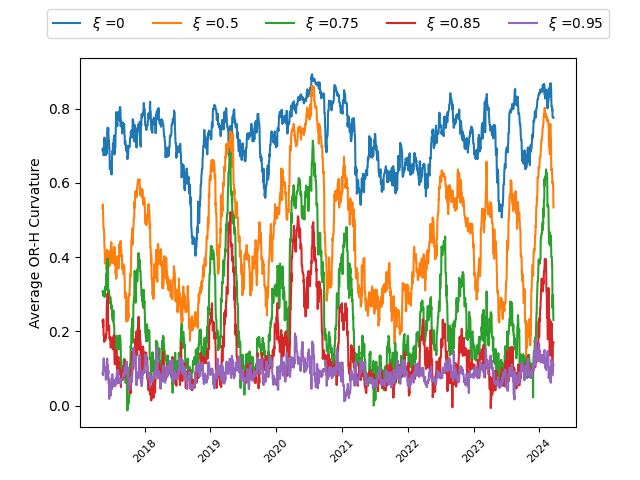}
    \caption{Changes on the method as $\xi$ varies, full time period 1997-2024, all the companies listed in NASDAQ Tech Sector (NDXT), $T = 132$.}
    \label{fig:SectorT132Xi}
\end{figure}
\begin{figure}[H]
    \centering
    \includegraphics[scale = 0.5 ]{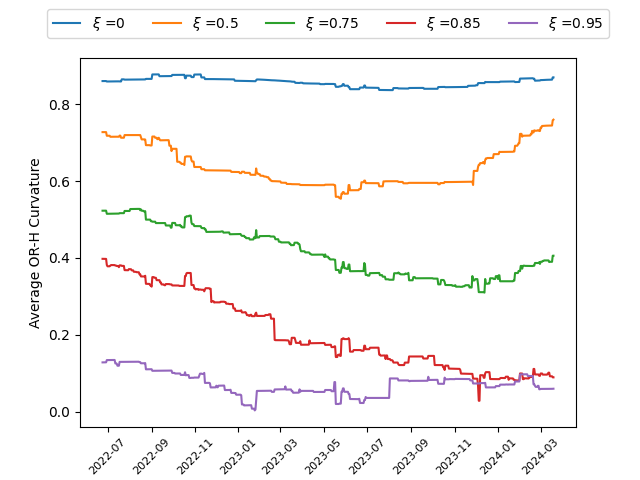}
    \caption{Changes on the method as $\xi$ varies, time period 2018-2024, all companies listed in NASDAQ Tech Sector (NDXT), $T = 2000$.}
    \label{fig:SectorT2000}
\end{figure}
Figure \ref{fig:SectorT2000} shows the complexity of choosing the parameters for the indicator. While we expect large $T$ ($2000$) to show a clearer picture, for the Tech Industry it seems to be ``too big''. In the $T = 2000$ case different values of $xi$ show different behaviours, while $T = 132$ successfully captures some peaks for most values of $\xi$. One possible explanation for this phenomena is that this sector is deemed to be more volatile, in some underlying sense, so large values of $T$ ``average out'' the local tendencies.

\subsection{Implementation of Ollivier-Ricci Gradient Sub-Sampling} \label{ORGSS}
As argued in Section \ref{MSTSection}, the $\ORicci$ curvature of the MST seems to have no direct relation to the $\ORicci$ curvature of the original graph. It is reasonable to study a sub-graph which maximizes or minimizes $\ORicci^{Net}$. Given a random initial vertex and a size $n$, one can generate a sub0graph with $n$ vertices and minimal $\ORicci^{Net}$ following the algorithm in \cite{ORGSSPaper}. Whether computing such an indicator performs better than the one presented in \cite{RicciFinance} is an interesting open question left for further research (with higher available computational power). 

\subsubsection{Algorithm for maximized $\ORicci$ curvature} \label{algorithmGSS}
\begin{algorithm}[H] 
To obtain a fragility indicator from $\ORicci$ gradient descent using $m$ nodes.
\begin{algorithmic}
\State Input: $m,T, \xi, \texttt{startDate}, \texttt{endDate}$
\For{$k \in \{1, \dots,  \texttt{endDate} - \texttt{startDate} - T \}$}
    \State Compute correlations matrix $\rho_{i,j}$ between stocks in period $[\texttt{startDate}, \texttt{startDate} + T]$.
    \State Compute cost function $D_{i,j} = \sqrt{2(1-\rho_{i,j})}$.
    \State Find via gradient descent (as in \cite{ORGSSPaper}) a sub-graph with $m$ nodes and maximal curvature.
    \State Add edge $xy$ to $G$ if $\rho_{x,y} \geq \xi$. 
    \State For every edge $ab$ compute the Ollivier-Ricci Curvature via \begin{equation}
        \kappa(a,b) = 1 - \frac{W_1(\mu_a,\mu_b)}{d(a,b)}.
    \end{equation}
    where, for a neighbor $v$ of node $a$ i.e. ($v \in N_a$) \begin{equation}
        \mu_a(v) = \frac{C_{a,v}}{\displaystyle \sum_{w \in N_a} C_{a,w}}
    \end{equation}
    \State and $d(a,b)$ is the (unweighted) HOP distance: counting the minimum number of steps in the shortest path of the extended graph. 
   \State Compute the average curvature $\ORicci^{net}$ by averaging $k$ over all edges in $G$.
\EndFor
\State Return the average curvature $\ORicci^{net}$
\end{algorithmic}
\end{algorithm}

We see in the next figure that the behaviour for small number $m$ is similar to the indicator of algorithm \ref{MainAlgorithm}. The technique of algorithm of section \ref{algorithmGSS} has the extra advantage of inheriting minimal curvature. In risk modelling and heavy-tail phenomena it is common to consider as best policy the one related to the worst case scenario. Minimum curvature sub-graphs do exactly that. We note that by computational restrictions our numerical results are not optimal. Two possible remedies are:
\begin{enumerate}
    \item Increase available computational power.
    \item Perform a \textit{stochastic} gradient descent on $\ORicci$ curvature. 
\end{enumerate}
We leave both of these directions for forthcoming research.

\subsubsection{Results for deterministic gradient sub-sampling}
\begin{figure}[H] \label{ORGSSPLOT} % "[t!]" placement specifier just for this example
\centering
\textbf{$\ORicci$ indicator from minimum curvature sub-graph, American market.}
\begin{subfigure}{0.48\textwidth}
\includegraphics[width=\linewidth]{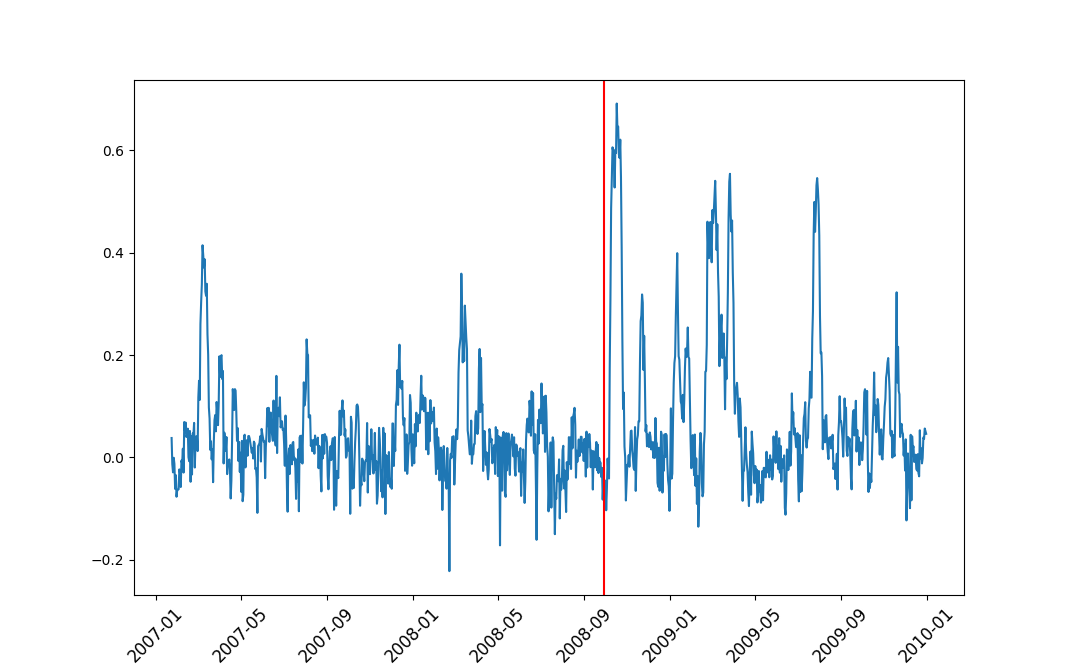}
\caption{$T = 22$ and $\xi = 0.85$.} 
\end{subfigure}\hspace*{\fill}
\begin{subfigure}{0.48\textwidth}
\includegraphics[width=\linewidth]{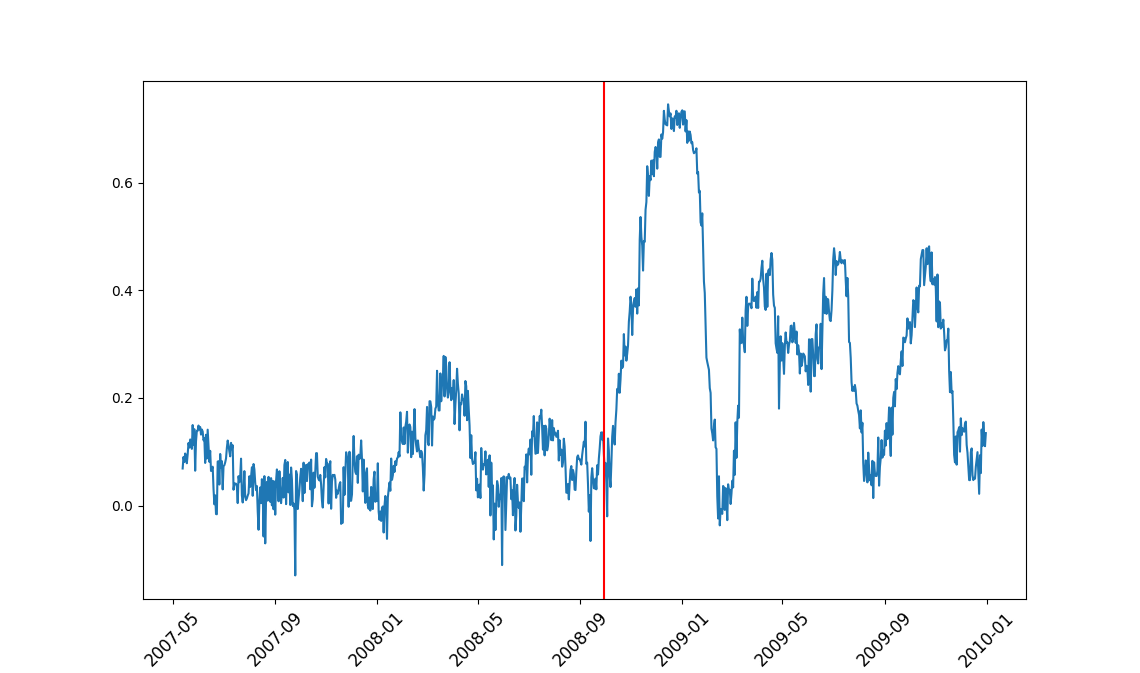}
\caption{$T = 132$ and $\xi = 0.75$.}
\end{subfigure}
\medskip
\begin{subfigure}{0.48\textwidth}
\includegraphics[width=\linewidth]{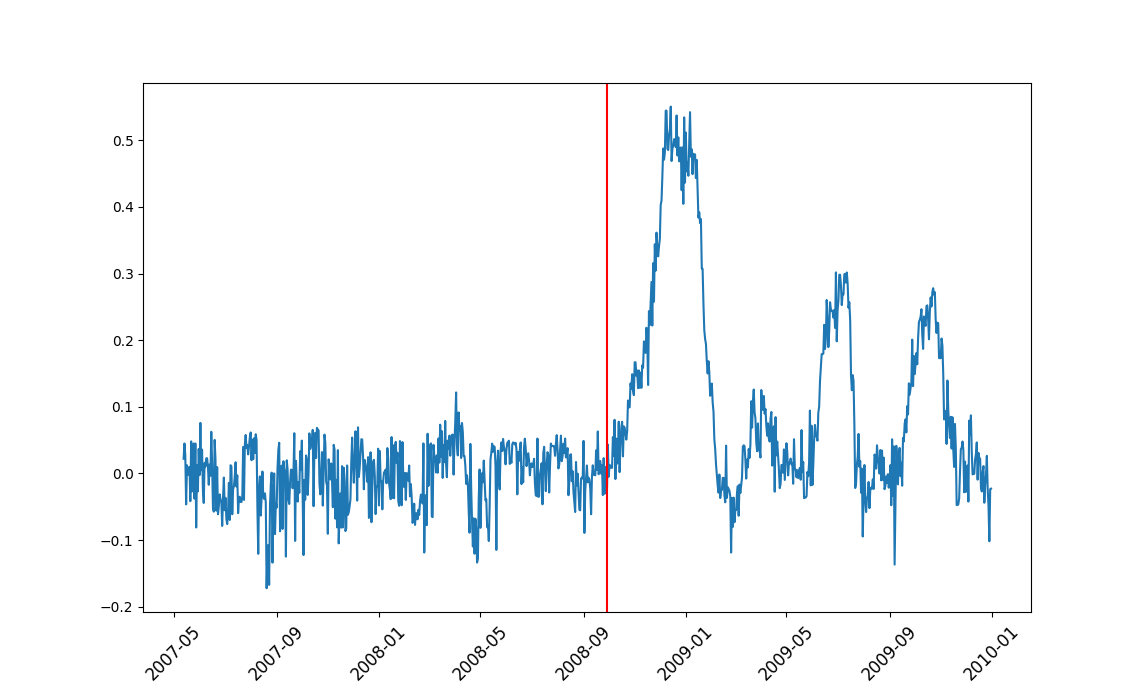}
\caption{$T = 132$ and $\xi = 0.85$.}
\end{subfigure}\hspace*{\fill}
\begin{subfigure}{0.48\textwidth}
\includegraphics[width=\linewidth]{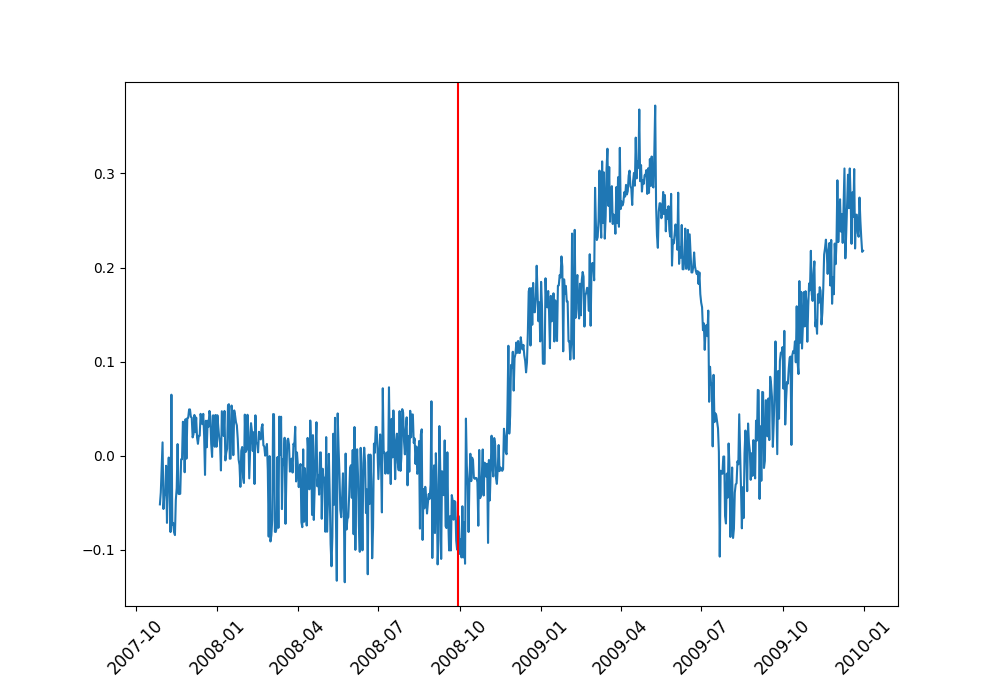}
\caption{ $T  = 300$ and $\xi = 0.85$.}
\end{subfigure}
\caption{$\ORicci$ curvature for a random sub-sample of $20$ stocks with sub-graph of minimum curvature using $20$ nodes and random initialization.}
\end{figure}

We remark that similar graphs using the \textit{maximum} curvature sub-graph yielded noisy graphs with no apparent trends.

\section{Conclusions and Further work}
In section \ref{AlgorithmSection} we posed the following questions, which we now answer:
\begin{enumerate}
    \item \textbf{Q}: How well does the indicator measure crisis? \\
    \textbf{Answer:} For certain values of $\xi$ (whose precise determination is not an easy task), the indicator measures times of financial crisis in length and depth. The indicator allows us to compare different crisis in terms of the hidden interconnections of financial agents.
    \item \textbf{Q}: What happens if one does not add high-value links using $\xi$ ? \\
    \textbf{Answer:} As observed in section \ref{XiSection}, if $\xi = 1$ (i.e. no link is added), the MST is not capable of showing the crashes.
    \item \textbf{Q:} What is the $\xi$- elasticity? How does one choose $\xi$? \\
    \textbf{Answer:} The choice of $\xi$ is fundamental to the analysis and depends not only on the distance function used, but also on the underlying financial market. The determination of $\xi$ should be empirical (or probably Bayesian) and heuristically $0.75 \leq \xi \leq 0.9$ can perform well. 
    \item \textbf{Q}: As $T$ grows does the time series regularize? \\
    \textbf{Answer:} In all cases bigger values of $T$ showed a clearer representation of the crisis. As $\ORicci^{Net}$ is supposed to be used as an ex-post tool rather than a predictor element, considering the biggest $T$ available is recommended.
    \item \textbf{Q}: As $T$ decreases do we approach white noise? \\
    \textbf{Answer:} Yes, as shown in section \ref{TSection}. This observation should alert a modeller that choosing a small value of $T$ (while desirable in applications) comes with bigger risk. 
    \item \textbf{Q}: What is the impact of using the MST? Are there other better subgraphs? \\
    \textbf{Answer:} If computationally plausible, one should attempt to use the closest representation of the original network graph. This is viable when studying interactions between sectors or small markets. In the case of large amounts of data, one can use MST, keeping in mind the examples in Section \ref{MSTSection}.
    \item \textbf{Q:} What happens for different cost functions $D_{i,j}$?  \\
    \textbf{Answer:} By relative stability of MST to weights (see \cite{StabilityOfMST}) the dependence on the cost is not crucial. Similar costs yield similar MSTs and when high-value links are added, the graphs behave similarly (the adding is done through correlations and not weights).
    \item \textbf{Q:} Is the average the most efficient way to measure risk? \\
    \textbf{Answer:} As explained in Example \ref{WeightIndependence}, there are many benefits of using the average Ricci curvature but there are interesting properties of other estimators which could be useful. The weight independence explained in Section \ref{MetricsSection} is the main motivation.
    \item \textbf{Q:} Is using the hop distance the only way to obtain the optimal structure or should it be used to compute the curvature too? \\
    \textbf{Answer:} Numerical experiments not included in this work have shown that using the weighted curvature is even less precise than the algorithm presented in section \ref{AlgorithmSection}.
\end{enumerate}

\subsection{Interpretation of the indicator}
Throughout this work we have shown that $\ORicci^{net}$ is a good indicator of crisis periods. The trends in the average Ollivier-Ricci curvature reliably indicate the size and duration of the instability periods of financial markets. For this reason, the indicator can be an important tool for economical and financial analysis. 

Due to it's dependence on the threshold variable $\xi$, as indicated in section \ref{XiSection}, one should be careful when using $\ORicci^{net}$ to predict short time tendencies of the market. In sections \ref{CountriesSection} and \ref{SectorsSection} we observed that the value of $\xi$ the modeller should use depends on the underlying financial market. Thus, it is reasonable to use a Bayesian approach to the determination the threshold. After calibration, $\xi$ should be understood as a hypothesis on which the inference of crisis from $\ORicci^{net}$ depends on. As long as this hypothesis is fully assumed, the indicator adds value to our understanding of the random dynamics of financial systems.

Note that Algorithm \ref{MainAlgorithm} depends on parameters $\xi$ and $D_{i,j}$ together with the use of the MST as the skeleton approach to the underlying network. The use of the MST is currently the most computationally feasible approach to markets with large amounts of data. If the market of interest is of small enough size one can use \eqref{OllivierRicci} directly from the constructed network, avoiding having to reduce to the MST and adding links via $\xi$.

\subsection{Taking away data could be crucially unfair}
One important observation about Algorithm \ref{MainAlgorithm} is that one should try to avoid dropping incomplete data. Say you want to analyze a market in a period of time where a crash happened, the crash may be a result of subsequent bankruptcy of entities with hidden connections. Computing the indicator after removing such entities misses the dependence on entities which can be fundamental to the crisis. For example, if several banks declare bankruptcy in the time period of study, one does not have complete data after their individual crashes, but it is these crashes that may explain the financial default of the system. Therefore, given a period of time, one should consider every entity which has participated in that period.

\subsection{Further research}
In this section we present the lines of investigation which remain open and we believe are of most interest.

\subsubsection{Relation between $\ORicci$ curvature and connectivity of a graph.}
We can understand Proposition \ref{firstbound} as a connection between curvature and connectivity. If a graph is highly connected, then the degrees of the vertices are very high, and the majority of edges should not significantly increase the hop distance if removed. In this case, \eqref{BoundWithDegrees} tells us that adding an edge will not result in a big change in curvature.

\subsubsection{The exponential moving average}
If we want to give the most recent data more influence, it is common practice in mathematical finance to use the exponential moving average. What happens if the exponential moving average is applied to the time series before any other calculation? Does it yield a better indicator in the hopes of avoiding the white noise in the $T \to 0$ limit? 

\subsubsection{The Beckman Problem}
Although equation \eqref{BoundWithDegrees} is a good connection between curvature and connectivity, there might be a stronger connection. The relation between the Beckmann and the Kantorovich Problem in the continuous case (see \cite[Theorem 4.6]{Santambrogio}) is well-known, whether or not similar relations hold in the discrete case (relating the max-flow min-cut problem would to the Wasserstein (Kantorovich) formulation of curvature) is unbeknownst to the authors. In the case where the weighted distance is used (different to the hop distance), \cite[Theorems 1,2]{WeightedTreeRicci}  are well known. We expect similar results for the hop-case.

\end{document}